\newif\ifsingle

\newif\ifFullversion
\Fullversiontrue 

\ifsingle
\documentclass[11pt,draftclsnofoot, onecolumn]{IEEEtran}		
\else		
\documentclass[9pt,final, twocolumn]{IEEEtran}
\fi

\usepackage{bbm}
\usepackage{times}
\usepackage{amsmath,dsfont}
\usepackage{amssymb,amsthm}
\usepackage{epsfig,verbatim}
\usepackage{ragged2e}

\usepackage{graphicx}
\usepackage{caption}
\usepackage{subcaption}
\usepackage{multicol}
\usepackage{cuted}

\usepackage{mwe}

\usepackage{setspace}
\usepackage{color}
\usepackage{cite}
\usepackage{epstopdf}
\usepackage{graphics}
\usepackage{accents}
\usepackage[nolist]{acronym}
\usepackage[bookmarks,colorlinks]{hyperref}
\usepackage{booktabs}
\usepackage{mathtools}
\usepackage{enumitem}
\usepackage{multirow}
\usepackage{booktabs}
\usepackage{blindtext}
\usepackage{epstopdf} 
\usepackage{optidef}

\usepackage{float}
\usepackage[ruled,linesnumbered]{algorithm2e}  
\SetKwInput{KwData}{\textbf{Initialize}} 
\usepackage{algpseudocode} 
\usepackage{titlesec}
\usepackage{graphicx}
\usepackage{lscape}
\usepackage{hyperref}
\usepackage{cleveref}
\usepackage{fancyhdr}
\fancypagestyle{github}{\fancyhf{}\fancyfoot[L]{The codes are available at \href{https://github.com/mahnoosh1/mahnoosh1-Collaborative-Learning-with-a-Drone-Orchestrator}{https://github.com/mahnoosh1/mahnoosh1-Collaborative-Learning-with-a-Drone-Orchestrator}}}

\newcommand{\norm}[1]{\left\lVert#1\right\rVert}

\newtheorem{theorem}{Theorem}

\definecolor{NewColor}{rgb}{0,0,0} 

\ifsingle

\else

\def\versegap{\hspace{1pt minus.5pt}}
\newcommand\versesep{\versegap\textbullet\versegap}
\fi 

\acrodef{adc}[ADC]{analog-to-digital convertor}
\acrodef{cs}[CS]{compressed sensing}
\acrodef{dtft}[DTFT]{discrete-time Fourier transform}
\acrodef{dnn}[NN]{neural network} 
\acrodef{csi}[CSI]{channel state information}
\acrodef{map}[MAP]{maximum a-posteriori probability}
\acrodef{snr}[SNR]{signal-to-noise ratio}
\acrodef{bs}[BS]{base station} 
\acrodef{iot}[IOT]{Internet of Things}
\acrodef{iov}[IOV]{Internet of Vehicles}
\acrodef{mimo}[MIMO]{multiple-input multiple-output}
\acrodef{mse}[MSE]{mean-squared error}
\acrodef{pdf}[PDF]{probability density function}
\acrodef{rv}[RV]{random variable}
\acrodef{ml}[ML]{machine learning}
\acrodef{mf}[MF]{matched filter}
\acrodef{fec}[FEC]{forward error correction}
\acrodef{rs}[RS]{Reed-Solomon}
\acrodef{lti}[LTI]{linear time-invariant}
\acrodef{wss}[WSS]{wide-sense stationary}
\acrodef{psd}[PSD]{power spectral density}
\acrodef{ser}[SER]{symbol error rate} 
\acrodef{ber}[BER]{bit error rate} 
\acrodef{sgd}[SGD]{stochastic gradient descent} 
\acrodef{isi}[ISI]{intersymbol interference}  
\acrodef{awgn}[AWGN]{additive white Gaussian noise} 
\acrodef{ut}[UT]{user terminal} 
\acrodef{mmw}[mmWave]{millimeter wave}
\acrodef{noma}[NOMA]{non-orthognal multiple access}
\acrodef{mac}[MAC]{mulitple access channel}
\acrodef{fl}[FL]{federated learning}
\acrodef{ct}[CT]{continuous-time}
\acrodef{uavs}[UAVs]{unmanned aerial vehicles}
\acrodef{uav}[UAV]{unmanned aerial vehicle}
\acrodef{dci}[DCI]{distributed collaborative intelligence}
\acrodef{psnr}[PSNR]{peak signal-to-noise ratio}

\SetKwInput{KwReq}{\textbf{Input}}

\IEEEoverridecommandlockouts 

\graphicspath{ {./images/} }

\title{Collaborative Learning with a Drone Orchestrator}

\author{Mahdi Boloursaz Mashhadi$^{*}$, \textit{Senior Member, IEEE}, Mahnoosh Mahdavimoghadam$^{*}$, \textit{Student Member, IEEE}, \\
Rahim Tafazolli$^{*}$, \textit{Senior Member, IEEE}, and Walid Saad$^{\dagger}$, \textit{Fellow, IEEE} \thanks{

Copyright (c) 2015 IEEE. Personal use of this material is permitted. However, permission to use this material for any other purposes must be obtained from the IEEE by sending a request to pubs-permissions@ieee.org. 

$^*$ M. Boloursaz Mashhadi, M. Mahdavimoghadam, and R. Tafazolli are
with 5GIC \& 6GIC, Institute for Communication Systems (ICS), University of Surrey, Guildford, United Kingdom (email: \{m.boloursazmashhadi, m.mahdavimoghadam, r.tafazolli\}@surrey.ac.uk). $^\dagger$ W. Saad is with the Bradley Department of Electrical and Computer Engineering, Virginia Tech, Arlington, VA, 22203, USA (email: walids@vt.edu).

This work was supported by a grant from the UK Department for Science, Innovation, and Technology under project TUDOR (Towards Ubiquitous 3D Open Resilient Network), and by the US National Science Foundation under Grant CNS-2114267. 

The codes are available at \href{https://github.com/DrMahdiBoloursazMashhadi/Collaborative-Learning-with-a-Drone-Orchestrator}{https://github.com/DrMahdiBoloursazMashhadi/
Collaborative-Learning-with-a-Drone-Orchestrator}.
}}

\vspace{-0.75cm}

\begin{document}
	
	\maketitle
	\pagestyle{empty}
	\thispagestyle{empty}
	
\begin{abstract}
\color{black} In this paper, the problem of drone-assisted collaborative learning is considered. In this scenario, \color{black} swarm of intelligent wireless devices train a shared neural network (NN) model with the help of a drone. Using its sensors, each device records samples from its environment to gather a local dataset for training. The training data is severely heterogeneous as various devices have different amount of data and sensor noise level. The intelligent devices iteratively train the NN on their local datasets and exchange the model parameters with the drone for aggregation. \color{black}For this system, \color{black}the convergence rate of collaborative learning \color{black} is derived while considering \color{black} data heterogeneity, sensor noise levels, and communication errors, \color{black}then, \color{black} the drone trajectory that maximizes the final accuracy of the trained NN \color{black} is obtained. The \color{black}proposed trajectory optimization approach is aware of both the devices data characteristics (i.e., local dataset size and noise level) and their wireless channel conditions, and significantly improves the convergence rate and final accuracy in comparison with baselines that only consider data characteristics or channel conditions. \color{black}{Compared to state-of-the-art baselines, the proposed approach achieves an average 3.85\% and 3.54\% improvement in the final accuracy of the trained NN on benchmark datasets for image recognition and semantic segmentation tasks, respectively. Moreover, the proposed framework achieves a significant speedup in training, leading to an average 24\% and 87\% saving in the drone's hovering time, communication overhead, and battery usage, respectively for these tasks.}      
\end{abstract}

\section{Introduction}\label{sec:intro}
Artificial intelligence (AI) will be a cornerstone of future wireless systems~\cite{chaccour2022less}. In particular, distributed learning will enable scattered intelligent connected devices to collaboratively train models on their locally available datasets without sharing raw data, thereby improving data privacy. \textcolor{black}{To collaboratively train a neural network (NN) model, intelligent devices can iteratively train the NN on their local datasets and communicate the NN parameters with an orchestrator which aggregates the models and sends it back to devices for further training \cite{DistLearning1, DistLearning2}}. In many applications, the intelligent devices are mobile and a drone could be an appropriate choice for orchestrating the training and aggregating the models due to its high mobility and flexible deployment. Drone-assisted collaborative learning has numerous use cases including \textcolor{black}{Artificial Intelligence of Things (AIoT)} scenarios in which numerous mobile IoT devices collaborate to train small AI models \cite{nguyen2021federated, UAVIoT, AIoT}, AI-assisted robotic swarm teleoperations in human inaccessible areas (e.g., survivor detection and rescue operation by collaborating robots post disaster), and Internet of Vehicles (IoV) applications \cite{IoV1, IoV2, IoV3} with autonomous vehicles collaboratively learning to drive in remote areas. 

\textcolor{black}{Very few studies \cite{zeng2017energy,pham2022energy,donevski2021federated,pham2021uav,do2021deep,wu2022distributed} have considered collaborative learning with the help of a drone. Due to limited battery, computation, and communication capacity of the devices and the drone, optimization of drone trajectory becomes necessary\cite{zeng2017energy}}. An energy-efficient federated learning (FL) framework is developed in \cite{pham2022energy} \textcolor{black}{for} minimizing the overall energy consumption of both drone and devices under various resource constraints. The authors in \cite{donevski2021federated}, derive the drone trajectory that minimizes device staleness, i.e. the largest difference in the epoch number between any two collaborating devices. \textcolor{black}{In \cite{pham2021uav}, a drone is leveraged to wirelessly power collaborating devices. The wireless resources and drone placement are optimized to achieve maximum power efficiency in this research.} A deep deterministic policy gradient (DDPG) algorithm is introduced in \cite{do2021deep} for joint drone placement and resource allocation so as to achieve energy-efficient FL. The authors in \cite{wu2022distributed} find the drones trajectory to maximize the average spectrum efficiency. In \cite{Merouane}, the authors perform trajectory optimization to minimize a weighted sum of energy and time consumption for UAV-enabled federated learning using a deep reinforcement learning approach.\color{black} The above studies \cite{zeng2017energy,pham2022energy,donevski2021federated,pham2021uav,do2021deep,wu2022distributed, Merouane} do not take into account the convergence speed and final performance of the trained NN model. This is despite the fact that the convergence speed directly affects the battery usage and total hovering time of the drone. In \cite{Joint}, the authors derive the expected convergence rate of federated averaging in presence of channel noise and fading in a static wireless scenario without transmitter/receiver mobility, assuming noise-free local datasets at the devices. However, the performance of collaborative learning is impacted both by the noise and fading in the drone-device wireless channel, as well as the size and quality (e.g. sensor noise level) of the local datasets at the intelligent devices. 

The main contribution of this work is to derive the expected convergence rate of collaborative training while taking into account effects of all the above factors and formulate the drone trajectory optimization problem based on our convergence result. We optimize the drone trajectory to speed up the training while considering device movements and the drone velocity constraints. \color{black} In summary, our key contributions include:
\begin{itemize}
    \item \color{black}We derive the expected convergence rate and the final performance of drone-assisted collaborative learning taking into account both noise and fading over the device-drone wireless channel, as well as the size and quality of the local device datasets considering heterogeneous and noisy local datasets at the devices. Using this result, we demonstrate how the final accuracy of the trained model is affected by both the communication errors and sensor noise levels while the speed of convergence is mostly determined by the communication errors only.\color{black}
    
    \item \color{black}Using our convergence result and considering both static and mobile intelligent devices, we analytically derive the optimal drone placement and speed to maximize accuracy and performance of the trained NN.\color{black}
    
    \item \color{black}Using our convergence result, we find the optimal drone trajectory that maximizes the asymptotic performance of the trained NN model while satisfying the drone velocity constraints. We demonstrate the importance of considering the noisy device datasets, by showing a considerable training speedup and an average 1.6\% improvement in image recognition accuracy, in comparison with the case in which the drone trajectory is optimized assuming the static noise-free convergence results in \cite{Joint}.\color{black}
    
    \item \color{black}{Compared to existing baselines, our proposed approach achieves an average improvement of 3.85\% in the final accuracy of the trained NN, as well as a significant speedup in the training, leading to a 24\% saving in the drone's hovering time, communication overhead, and battery usage.}  
\end{itemize}

The rest of the paper is organized as follows. In \Cref{sec:Model}, we present the system model. \Cref{sec:Approach}, provides the problem formulation and our proposed approach. \Cref{sec:Results} presents the simulation results, and \Cref{sec:Conclusions} concludes the paper.

\color{black}
\textbf{Notations:} Boldface capital and lower-case letters represent matrices and vectors, respectively (e.g. $\boldsymbol{A}$ and $\boldsymbol{a}$). Calligraphic letters denote sets (e.g. $\mathcal{A}$). The matrix transpose expressed by $(\cdot)^{\dagger}$. Furthermore, $\lVert . \rVert$ is the norm value, and $\mathbb{E}\{\cdot\}$ is the expectation operator. $U\left[ .,.\right]$ and $N(.,.)$ represent the uniform and normal distributions, respectively. $\boldsymbol{I}$ denotes the identity matrix.
\color{black}
\section{System Model}\label{sec:Model}
\begin{figure}[t]
    \centering
    \includegraphics[scale=.3]{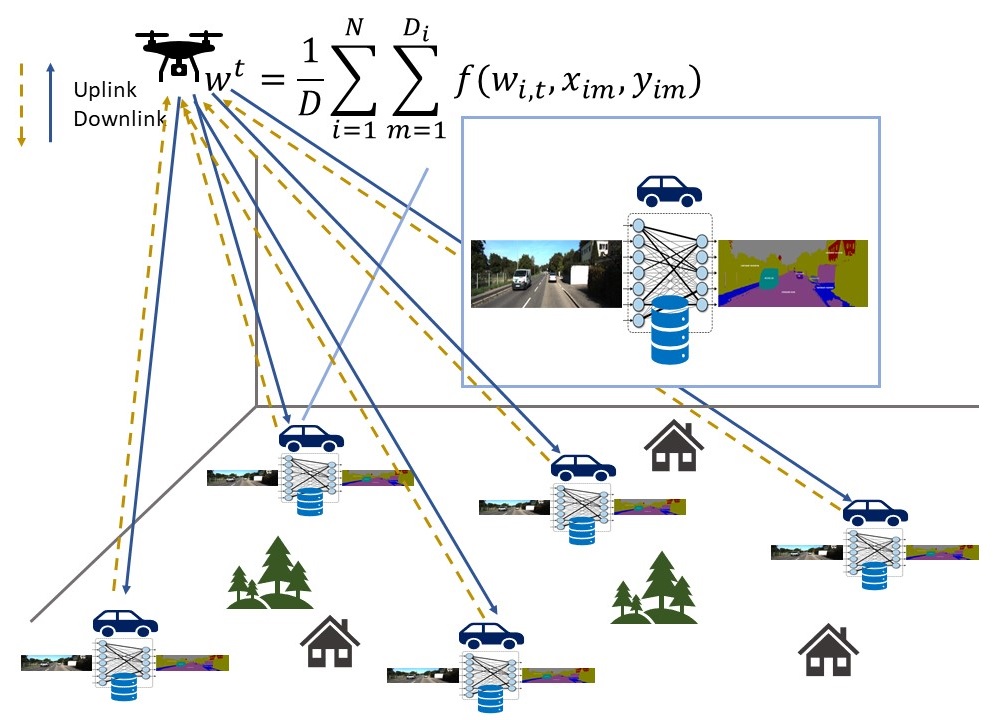}
    \caption{Collaborative learning orchestrated by a drone.}
    \label{fig:system_model}
\end{figure}
\textcolor{black}{We consider a set of $N$ intelligent wireless devices} (e.g., driverless cars, AIoT devices, rescue robots, etc.) collaborating to train a shared NN model with the help of a drone, as shown in Fig. ~\ref{fig:system_model}. Each device $i$ has a local dataset captured by its own sensors \textcolor{black}{(e.g., cameras, GPS, motion/temperature/light sensors, etc.)}, and denoted by $\mathcal{D}_i = \{\boldsymbol {x}_{im}+\boldsymbol {n}_{im}, y_{im}\}$, where $m = 1, \hdots , D_i=|\mathcal{D}_i|$, and $\boldsymbol {x}_{im}$ and $y_{im}$ are the model input features and output labels, respectively, and $\boldsymbol {n}_{im}$ is a zero-mean Gaussian sensor noise with variance $\sigma^2_i$, i.e., $\boldsymbol {n}_{im} \sim \mathcal{N}(0, \sigma^2_i \boldsymbol {I})$. \color{black}The data quality at device $i$ is heterogeneously determined by $\sigma_i^2$.\color{black} The training loss function is $f(\boldsymbol {w}, \boldsymbol {x}_{im}+\boldsymbol {n}_{im}, y_{im})$, where $\boldsymbol {w}$ is the vector of NN model parameters. The average local loss at device $i$ is given by $   F_i(\boldsymbol {w_t}) = 1/D_i \sum_{m=1}^{D_i} f(\boldsymbol {w_t}, \boldsymbol {x}_{im}+\boldsymbol {n}_{im}, y_{im})$ and the global loss over $N$ devices is $F(\boldsymbol {w_t}) = \sum_{i=1}^{N} \frac{D_i}{D} F_i(\boldsymbol {w_t})=\frac{1}{D}\sum_{i=1}^{N} \sum_{m=1}^{D_i} f(\boldsymbol {w_{i,t}}, \boldsymbol {x}_{im}+\boldsymbol {n}_{im}, y_{im})$, where $D = \sum_{i=1}^{N} D_i$ is the total size of the entire dataset.

Collaborative training is orchestrated by the drone and the objective is to minimize the global loss $F(\boldsymbol {w})$. \color{black}Training is iterative, where each device $i$ at aggregation round $t$ performs a stochastic gradient descent update on its local dataset according to $\boldsymbol {w}_{i,t}=\boldsymbol {w}_{t-1}-\lambda \nabla F_i(\boldsymbol {w}_{t-1})$, $\lambda$ being the learning rate, and then communicates $\boldsymbol {w}_{i,t}$ to the drone. \color{black} The drone then aggregates all device models according to $\boldsymbol {w}_t=\frac{\sum_{i=1}^{N} D_i \boldsymbol {w}_{i,t}}{D}$, and \textcolor{black}{sends $\boldsymbol {w}_{t+1}$} back to the devices for further training. This process carries on iteratively until convergence. Although we consider the FedAvg \cite{FedAvg} algorithm in our framework, an extension of the results to other collaborative learning algorithms \cite{FedBoost}, \cite{FedNova} is straightforward.

We consider each local model $\boldsymbol {w}_i$ is transmitted as a single packet in the uplink using a cyclic redundancy check (CRC) mechanism to detect errors in the received models at the drone. \color{black}Whenever the received local model contains errors, i.e. the CRC is wrong, the drone will simply exclude the erroneous model from the global model aggregation and will not ask the corresponding
device to re-transmit. This is because due to the iterative nature of federated training, updates from that specific device will still take part in many future rounds of model aggregations, thereby still sufficiently participating in training of the global model. Simulations show that in many scenarios, it is not worth asking for re-transmissions due to the iterative nature of federated training.\color{black} ~Denoting the CRC checksum by $C(\boldsymbol {w}_{i,t}) \in \{0,1\}$, the drone will update the global FL model according to $\boldsymbol {w_t}=\frac{\sum_{i=1}^{N} D_i C(\boldsymbol {w}_{i,t}) \boldsymbol {w}_{i,t}}{\sum_{i=1}^{N} D_i C(\boldsymbol {w}_{i,t})}$ \cite{Joint}. The average packet error rate for device $i$ is given by $e_{i,t}=1-\exp{\left(-\frac{\theta B N_0}{\mathbb{E}|h_{i,t}|^2 \rho_i}\right)}$,
where $h_{i,t}$ is the channel gain between device $i$ and the drone, $B$ is the bandwidth, $N_0$ is the noise power spectral density, and $\theta$ is a waterfall threshold \cite{Waterfall}.

\color{black}The channel gain between the drone and device $i$ is expressed as $|h_{i,t}|^2=\pi_{i,t}\delta_{i,t}$, in which $\delta_{i,t}$ denotes the small-scale fading  term with $\mathbb{E}\{\delta_{i,t}\}=\nu_i$ and $\pi_{i,t}$ is the path loss. The average path loss depends on the probability of a line of sight (LoS) link between the drone and device $i$ at aggregation round $t$. We have $\mathbb{E}\{\pi_{i,t}\}=(\frac{s}{4 \pi f_c })^2d_{i,t}^{-\alpha}(\pi_{i,t}^{LoS} \zeta_{i,t}+\pi_{i,t}^{NLoS} (1-\zeta_{i,t}))$, where $d_{i,t}$ is the distance between the drone and device $i$ at aggregation round $t$, $f_c$ is the carrier frequency, and the constants $s$ and $\alpha$ represent the speed of light, and the path loss exponent. Here, $\pi_{i,t}^{LoS}$ and $\pi_{i,t}^{NLoS}$ represent the additional losses for the LoS and non-LoS links, respectively, and $\zeta_{i,t}$ is the probability of having a LoS link between the drone and device $i$ at aggregation round $t$. This probability can be given by $\zeta_{i,t}=\frac{1}{1 + a \exp{(-b (\kappa_{i,t} - a))}}$, where $\kappa_{i,t}$ is the elevation angle between the drone and device $i$ at time $t$, and $a$, $b$ are environment-related parameters \cite{Mozaffari, Alzenad, Elzanaty}. According to experimental measurements presented in \cite{LTE}, for moderate altitudes (less than 100 meters), $\zeta_{i,t} \approx 1$ and we use this assumption hereinafter. The error rate $e_{i,t}(\boldsymbol {p}^t, \boldsymbol {p}^t_i)=1-\exp{\left(-\frac{\theta B N_0}{(\frac{s}{4 \pi f_c })^2d_{i,t}^{-\alpha} \pi_{i,t}^{LoS} \nu_i \rho_i}\right)}$ is hence a function of the drone $\boldsymbol {p}^t=(x^t, y^t, H)$ and device positions $\boldsymbol {p}_i^t=(x^t_i, y^t_i, 0)$ at aggregation round $t$, where $d_{i,t}=((x^t-x^t_i)^2+(y^d-y^t_i)^2+H^2)^{1/2}$. In the rest of this work, at times, we use $e_{i,t}$ for notation simplicity but we keep in mind that this is a function of the drone and device locations.\color{black}

As the model parameters are iteratively exchanged between devices and the drone through the wireless channel, communication errors will affect the performance of collaborative learning. The drone position during training determines the channel gain and the communication error rate between the drone and the devices, thereby affecting the performance of collaborative training. On the other hand, performance of collaborative learning depends on the amount and quality of data locally available at the devices, e.g., a higher sensor noise at a device will degrade quality of its local dataset. Hence, the optimal trajectory for the drone should be determined by a joint optimization that considers both the wireless factors, i.e., the channel gain and signal-to-noise ratio (SNR), and the learning parameters, i.e., the amount and quality of local data. The goal of this paper is to optimize the drone trajectory to achieve the fastest learning and the best final performance for the trained NN model.

\vspace{5mm}

\section{Trajectory Optimization}\label{sec:Approach}




\color{black}
To derive the expected convergence rate of the collaborative learning algorithm, we make the following assumptions. These assumptions can be satisfied by several widely used loss functions such as cross-entropy, mean squared error, and logistic regression \cite{Lipschitz}. 

\vspace{2mm}
\noindent I) We assume that the gradient $\nabla F(\boldsymbol {w})$ of $F(\boldsymbol {w})$ is uniformly Lipschitz continuous with respect to $\boldsymbol {w}$ \cite{Lipschitz}, i.e. $ \|\nabla F(\boldsymbol {w}_{t+1})-\nabla F(\boldsymbol {w}_{t})\| \le L \|\boldsymbol {w}_{t+1}-\boldsymbol {w}_{t}\|,$ where $L$ is a positive constant and $\|\boldsymbol{w}_{t+1}-\boldsymbol{w}_{t}\|$ is the norm of $\boldsymbol{w}_{t+1}-\boldsymbol{w}_{t}$.

\vspace{2mm}
\noindent II) We assume that $F(\boldsymbol{w})$ is strongly convex with positive parameter $\mu$, such that $F(\boldsymbol{w}_{t+1}) \ge F(\boldsymbol{w}_t)+(\boldsymbol{w}_{t+1}-\boldsymbol{w}_{t})^T \nabla F(\boldsymbol{w}) +\mu/2\|\boldsymbol{w}_{t+1}-\boldsymbol{w}_{t}\|^2$.

\vspace{2mm}
\noindent III) We assumed that $F(\boldsymbol{w})$ is twice continuously differentiable. We get $\mu \boldsymbol{I} \le \nabla^2 F(\boldsymbol{w}) \le L \boldsymbol{I}$.

\vspace{2mm}
\noindent IV) We assume $\nabla_{\boldsymbol{wx}}^2 f(\boldsymbol{w}_t, \boldsymbol{x}_{im}, y_{im})^T \nabla_{\boldsymbol{wx}}^2 f(\boldsymbol{w}_t, \boldsymbol{x}_{im}, y_{im}) \le \eta \boldsymbol{I}$.

\vspace{2mm}
\noindent V) Finally, we assume that $\|\nabla f(\boldsymbol{w}_t, \boldsymbol{x}_{im}, y_{im})\|^2 \le c_1 + c_2 \|\nabla F(\boldsymbol{w}_t)\|^2$ with $c_1, c_2 > 0$.

\vspace{2mm}

\color{black}{
\begin{theorem} \label{thm:Thrm}
Given the learning rate 
$\lambda=1/L$, we have

\begin{align}\label{Theorem}
    &\mathbb{E}(F(\boldsymbol{w}_{t+1})-F(\boldsymbol{w}^*)) \le \Phi_t \mathbb{E}(F(\boldsymbol{w}_{t})-F(\boldsymbol{w}^*))+(J_t+K_t),
\end{align}
where $\boldsymbol{w}^*$ and $\boldsymbol{w}_{t}$ are the vector of optimal NN model parameters, and the global NN parameters at aggregation round $t$, respectively. The error terms $J_t=\frac{2c_1}{L D} \sum_{i=1}^{N} D_i e_{i,t}$ and $K_t=\frac{\eta M}{2L D^2} \sum_{i=1}^{N} D_i (1-e_{i,t}) \sigma^2_i$ are due to the packet errors and sensor noise, and $\Phi_{t}=1-\frac{\mu}{L}+\frac{4\mu c_2}{LD}\sum_{i=1}^N D_i e_{i,t}$. 

\end{theorem}}

\begin{proof}
    See APPENDIX. \ref{app:xx}
\end{proof}
\color{black}
The above theorem provides a bound on the expected improvement of the training loss in each aggregation round taking into account both packet errors $e_{i,t}$ and sensor noise levels $\sigma^2_i$. Iterative application of (\ref{Theorem}) over a series of $T$ aggregation rounds, gives a bound on the expected convergence speed and final performance of the trained NN. Assuming $\{\Phi_t\}_{t=1}^{T} < 1$, the speed of convergence is determined only by $\{\Phi_t\}_{t=1}^{T}$ values, while the expected final performance of the trained NN is determined by the interplay of the collective effects from $J_t$, $K_t$, and $\Phi_t$ terms over $T$ aggregation rounds. In an ideal scenario without packet errors and sensor noises, i.e. assuming $e_{i,t}=0$ and $\sigma_i^2=0$, the expected training loss would asymptotically converge to the optimum loss value, i.e. $\lim_{T \to \infty} {\mathbb{E}(F(\omega_{T+1}))}=F(\omega^*)$, while the packet errors and sensor noises have a destructive effect on the convergence speed and final performance of the trained NN.

In comparison with \cite{Joint}, \textbf{Theorem 1} shows the separate effects of the sensor noise levels on the accuracy of the trained NN, allowing noise-aware optimization of the drone trajectory. The packet errors increase $\Phi_t$ values, thereby slowing down the convergence. The sensor noises seem to have no effect on the convergence speed, but degrade the final performance of the trained NN by introducing the additional error terms $K_t$ values. To speed up convergence and improve the final performance of the NN, the drone should position itself closer to devices with larger local datasets and less noisy data. In the following, we study the optimum placement and trajectory of the drone, based on the above insights.
\color{black}

\subsection{Drone placement for stationary devices}\label{placement}
Assuming stationary devices, a fixed position can be optimized for the drone. In this stationary case, we have $\boldsymbol{p}^t=\boldsymbol{p}^d=(x^d,y^d)$, $e_{i,t}=e_i$, and $\Phi_t=\Phi, \forall t \in \{1,\dots,T\}$. Hence, by applying (\ref{Theorem}) recursively, we get:
\begin{align}\label{2}
    &\mathbb{E}(F(\boldsymbol{w}_{T+1})-F(\boldsymbol{w}^*)) \le \Phi^T \mathbb{E}(F(\boldsymbol{w}_{0})-F(\boldsymbol{w}^*)) \\ \nonumber
    &+\left ( \frac{2c_1}{L D} \sum_{i=1}^{N} D_i e_i+\frac{\eta M}{2L D^2} \sum_{i=1}^{N} D_i (1-e_i) \sigma^2_i \right )\frac{1-\Phi^{T}}{1-\Phi}.
\end{align}

\begin{algorithm} [t] 
		\caption{\color{black} Optimal Placement for Drone-assisted Collaborative Learning with Stationary Devices}
  \color{black}
		\label{alg:aloc1}
		\KwReq{$\left\{(x_i,y_i)\right\}_{i=1}^N$, $\left\{D_i\right\}_{i=1}^N$, $\left\{A_i\right\}_{i=1}^N$, $\left\{\rho_i\right\}_{i=1}^N$, $\theta$, $B$, $N_0$, $\alpha$, $\eta$, $c_1$, $L$, $M$, $H$, $\delta$.}
		\KwData{$k \!=\! 1$, $\mathbf{p}^d_0 \!=\! (0,0)$, $\mathbf{p}^d_1 \!=\! \left( \frac{\sum_{i=1}^{N}D_ix_i}{\sum_{i=1}^{N}D_i}, \frac{\sum_{i=1}^{N}D_iy_i}{\sum_{i=1}^{N}D_i} \right)$.}
		\While{$\|\mathbf{p}^d_k-\mathbf{p}^d_{k-1}\| > \delta$ }{
			{-Calculate the partial derivatives (PDs) (\ref{Derivative2}), (\ref{Derivative3}) at ${\mathbf{p}}^d \!=\! {\mathbf{p}}^d_k$,\\ 
			-Using PDs, derive linearizations $f_k(\boldsymbol{p}^d) \! \approx \! \boldsymbol{c}^{\dagger}_k \boldsymbol{p}^d \! + \! \beta_k$ and $g_k(\boldsymbol{p}^d) \! \approx \! \boldsymbol{d}^{\dagger}_k \boldsymbol{p}^d \! + \! \gamma_k$ at ${\mathbf{p}}^d \!=\! {\mathbf{p}}^d_k$,\\
			-Use the Charnes-Cooper’s transformations $\boldsymbol{q}_k \! = \! \frac{\boldsymbol{p}^d}{\boldsymbol{d}^{\dagger}_k \boldsymbol{p}^d+\gamma_k}$ and $z_k \! = \! \frac{1}{\boldsymbol{d}^{\dagger}_k \boldsymbol{p}^d+\gamma_k}$ to get optimization (\ref{eq:harnes-Cooper}) at ${\mathbf{p}}^d \!=\! {\mathbf{p}}^d_k$,\\
			-Solve the resulting linear programming (\ref{eq:harnes-Cooper}) to get ${\mathbf{p}}^d_{k+1}$,\\
			-$k=k+1$,}}
	\KwOut{Optimal drone location ${\mathbf{p}^*}^d={\mathbf{p}}^d_k$.}
\end{algorithm}

Given $\Phi<1,\versesep T \rightarrow \infty$, the first term on the right-hand side of (\ref{2}) vanishes and the asymptotic learning performance is determined by the term $(\frac{2c_1}{L D} \sum_{i=1}^{N} D_i e_i+\frac{\eta M}{2L D^2} \sum_{i=1}^{N} D_i (1-e_i) \sigma^2_i)\frac{1}{1-\Phi}$. \color{black} Hence, the optimal drone position can be derived by solving
\color{black}
\begin{mini}|s|[0]
    {(x^d,y^d)}{\left(\frac{2c_1}{L D} \sum_{i=1}^{N} D_i e_i+\frac{\eta M}{2L D^2} \sum_{i=1}^{N} D_i (1-e_i) \sigma^2_i\right)\frac{1}{1-\Phi}}
    {}
    {\label{eq:ATL_placement}}{}
    \addConstraint{}{\Phi<1.}
\end{mini}
Optimization (\ref{eq:ATL_placement}) is a fractional programming problem which we solve using the Charnes-Cooper's \cite{Schaible1974ParameterfreeCE} variable transformation technique. To this end, we first obtain a linear approximation of the numerator $f(x^d,y^d)=\frac{2c_1}{L D} \sum_{i=1}^{N} D_i e_i+\frac{\eta M}{2L D^2} \sum_{i=1}^{N} D_i (1-e_i) \sigma^2_i$ and denominator $g(x^d,y^d)=1-\Phi=\frac{\mu}{L}-\frac{4\mu c_2}{LD}\sum_{i=1}^N D_i e_{i}$ functions in terms of the decision variables $(x^d,y^d)$. It is straightforward to derive the following


\begin{align}\label{Derivative2}
    \frac{\partial f}{\partial x^d}&=\sum_{i=1}^{N} \frac{D_i}{D}(\frac{2c_1}{L}-\frac{\eta M \sigma_i^2}{2L D})(\frac{\theta B N_0 \alpha }{\rho_i A_i}) d_i^{\alpha-2} \\ \nonumber
    &\times \exp{\left(-\frac{\theta B N_0 d_i^{\alpha}}{A_i \rho_i}\right)} (x^d-x_i), \\ \nonumber
    \frac{\partial f}{\partial y^d}&=\sum_{i=1}^{N} \frac{D_i}{D}(\frac{2c_1}{L}-\frac{\eta M \sigma_i^2}{2L D})(\frac{\theta B N_0 \alpha }{\rho_i A_i}) d_i^{\alpha-2} \\ \nonumber
    &\times \exp{\left(-\frac{\theta B N_0 d_i^{\alpha}}{A_i \rho_i}\right)} (y^d-y_i),
\end{align}
where $A_i=(\frac{s}{4\pi f_c})^2 \pi_{i,t}^{LoS} \nu_i$, and $\frac{\partial (.)}{\partial (.)}$ denotes the partial derivative. We also have
\begin{align}\label{Derivative3}
    \frac{\partial g}{\partial x^d}&=\sum_{i=1}^{N} (\frac{4\mu c_2 D_i}{L D})(\frac{\theta B N_0 \alpha }{\rho_i A_i}) d_i^{\alpha-2} \exp{\left(-\frac{\theta B N_0 d_i^{\alpha}}{A_i \rho_i}\right)} (x^d-x_i), \\ \nonumber
    \frac{\partial g}{\partial y^d}&=\sum_{i=1}^{N} (\frac{4\mu c_2 D_i}{L D})(\frac{\theta B N_0 \alpha }{\rho_i A_i}) d_i^{\alpha-2} \exp{\left(-\frac{\theta B N_0 d_i^{\alpha}}{A_i \rho_i}\right)} (y^d-y_i).
\end{align}

\color{black} Now note that we then derive the linear approximation of $f(\boldsymbol{p}^d)$ and $g(\boldsymbol{p}^d)$ in the vicinity of geometrical weighted centroid $(x^c, y^c)=\left( \frac{\sum_{i=1}^{N}D_ix_i}{\sum_{i=1}^{N}D_i}, \frac{\sum_{i=1}^{N}D_iy_i}{\sum_{i=1}^{N}D_i} \right)$. This is due to the fact that we expect and also observe in simulations that the optimal drone position is normally not too far from the geometrically weighted centroid, closer to any device with a larger local dataset. Now with the linearization $f(\boldsymbol{p}^d)\approx \boldsymbol{c}^{\dagger} \boldsymbol{p}^d+\beta$, $g(\boldsymbol{p}^d) \approx \boldsymbol{d}^{\dagger} \boldsymbol{p}^d+\gamma$, and using the Charnes-Cooper’s variable transformations $\boldsymbol{q}=\frac{\boldsymbol{p}^d}{\boldsymbol{d}^{\dagger} \boldsymbol{p}^d+\gamma}$ and $z=\frac{1}{\boldsymbol{d}^{\dagger} \boldsymbol{p}^d+\gamma}$, optimization (\ref{eq:ATL_placement}) is equivalent to the following linear programming problem 
\begin{mini}|s|[0] 
    {\boldsymbol{p}^d}{\boldsymbol{c}^{\dagger} \boldsymbol{q}^d+\beta z}
    {}
    {\label{eq:harnes-Cooper}}{}
    \addConstraint{}{\boldsymbol{d}^{\dagger} \boldsymbol{q}^d+\gamma z = 1}
    \addConstraint{}{z>0.}
\end{mini}
\color{black}To find the optimal drone position for stationary devices, we take an iterative approach starting from the weighted centroid and iteratively perform linearizations (\ref{Derivative2}) and (\ref{Derivative3}), as well as Charnes-Cooper’s fractional programming (\ref{eq:harnes-Cooper}) till convergence. This approach solves optimization (\ref{eq:ATL_placement}) in a fast and computationally efficient way and the simulation results show that we usually achieve convergence with very few (i.e., usually less than 10) iterations. \color{black}\textbf{Algorithm} \ref{alg:aloc1} presents our proposed drone placement algorithm for collaborative learning with stationary devices. We would like to emphasize here that this is an asymptotic solution assuming a sufficiently large number of aggregation rounds, i.e. $T \to \infty$. Without this assumption, the optimum drone placement would depend on the first term on the right hand side of (\ref{2}), and consequently on the choice of the initial NN weights, i.e. $\omega_0$. Although, this solution is based on iterative linearization of (\ref{eq:ATL_placement}) and thereby suboptimal, the simulation results show sufficient quality of the solution, leading to a significant performance improvement in comparison with other baselines. 
\color{black}

\begin{algorithm} [t] 
		\caption{\color{black} Trajectory Optimization for Drone-assisted Collaborative Learning with Moving Devices}
  \color{black}
		\label{alg:aloc2}
		\KwReq{$\left\{(\Tilde{x}_i^{0},\Tilde{y}_i^{0})\right\}_{i=1}^N$, $\left\{D_i\right\}_{i=1}^N$, $\left\{A_i\right\}_{i=1}^N$, $\left\{\rho_i\right\}_{i=1}^N$, $M$, $\delta$, $\left\{\left\{(v^t_{x_i}, v^t_{y_i})\right\}_{i=1}^N\right\}_{t=1}^{T}$, $v_{max}$, $H$, $\theta$, $B$, $N_0$, $\alpha$, $\eta$, $c_1$, $L$.}
		\KwData{Using \textbf{Algorithm 1} optimize the initial drone location $({x^*}^{0},{y^*}^{0})$ at $t=0$.}
		\For{$t=1, \hdots, T$ }{
			{-Calculate the projected speed by ${{{{v^*}}_x^t}}=\min\{{{\hat{v}_x^t}}, v_{x, max}\}$ and ${{{{v^*}}_y^t}}=\min\{{{\hat{v}_y^t}}, v_{y, max}\}$ using (\ref{Derivative5}),\\ 
			-Set the next trajectory point as: $({x^*}^{t},{y^*}^{t})=({x^*}^{t-1}+{{{{v^*}}_x^t}},{y^*}^{t-1}+{{{{v^*}}_y^t}})$,}}
	\KwOut{Optimal drone trajectory $\left\{({x^*}^{t},{y^*}^{t})\right\}_{t=1}^T$.}
\end{algorithm}

\subsection{Trajectory optimization for moving devices}\label{trajectory}
For the general case of moving devices, and by applying (\ref{Theorem}) recursively, we get:
\begin{align}\label{Corollary1}
    &\mathbb{E}(F(\boldsymbol{w}_{T+1})-F(\boldsymbol{w}^*)) \le \prod_{t=1}^T \Phi_t \mathbb{E}(F(\boldsymbol{w}_{0})-F(\boldsymbol{w}^*)) \\ \nonumber
    &+\sum_{t=1}^{T-1}(J_t+K_t) \prod_{\tau=t+1}^{T} \Phi_{\tau}+(J_T+K_T).
\end{align}
Similarly, with $\Phi_t < 1$ and $ T \rightarrow \infty$, the first term on the right-hand side of (\ref{Corollary1}) vanishes and the asymptotic learning performance is determined by the residual term which itself depends on the drone trajectory. We define the asymptotic trajectory loss (ATL) as $\sum_{t=1}^{T-1}(J_t+K_t) \prod_{\tau=t+1}^{T} \Phi_{\tau}+(J_T+K_T)$, and use it as a surrogate of the final performance of the trained model for trajectory optimization.

Considering the asymptotic performance of the training (\ref{Corollary1}), the drone trajectory optimization reduces to minimizing the ATL, given by 
\begin{mini}|s|[0]
    {{{\{\boldsymbol{p}^1 \dots \boldsymbol{p}^{\frac{T}{\kappa}+1}\}} }}{\sum_{t=1}^{T-1}(J_t+K_t) \prod_{\tau=t+1}^{T} \Phi_{\tau}+(J_T+K_T)}
    {}
    {\label{eq:ATL_Optimization}}{}
    \addConstraint{}{\Phi_t < 1,}{\quad\forall t=1,\hdots,T}
    \addConstraint{}{\norm{\boldsymbol{p}^{t'+1}-\boldsymbol{p}^{t'}}}{}{\le v_\textrm{max},}{\quad\forall t'=1,\hdots,\frac{T}{\kappa}}
    \addConstraint{}{\boldsymbol{p}^1=\boldsymbol{p}^{\frac{T}{\kappa}+1}},
\end{mini}
\color{black}where $\boldsymbol{p}^1 \dots \boldsymbol{p}^{\frac{T}{\kappa}+1}$ denote the drone trajectory points and we assume the drone remains at each trajectory point for $\kappa$ aggregation rounds and $\frac{T}{\kappa}$ is the total number of trajectory points.\color{black} The optimization is carried out over a horizon of $T$ aggregation rounds, and the added constraints enforce the maximum possible speed of the drone denoted by $v_\textrm{max}$, and ensure that the drone returns to its starting point (i.e., the drone's charging station) after training. Note that the ATL minimized in (\ref{eq:ATL_Optimization}), is a function of the drone trajectory through the device channel gains $\mathbb{E}|h_{i,t}|^2$ and packet error rates $e_{i,t}$ for each trajectory position $\boldsymbol{p}^{t'}$. Optimization of the drone trajectory in the most general case (\ref{eq:ATL_Optimization}), is a constrained non-linear programming problem in terms of the trajectory points $\boldsymbol{p}^{t'}$, which we solve using sequential quadratic programming (SQP) \cite{Boyd}. 

\color{black} Denoting the total time required for each aggregation round by \textit{aggregation time}, i.e. $T_{aggr}$, and the time required for the drone to traverse in between the consecutive trajectory points by $T_{trav}$. We assume the \textit{aggregation time} to be negligible in comparison with the traverse time, i.e. $T_{trav} >> T_{aggr}$. This assumption is valid in typical scenarios where sufficient computation and communication resources is available for collaborative learning at the devices, and there is a rich literature that consider minimization of the aggregation time and the corresponding trade-offs \cite{DistLearning1,Joint,AoI}. With this assumption, the total training time is dominated by $T_{trav}$ and can be approximated by the total number of the trajectory points multiplied by the traverse time, i.e. $\frac{T}{\kappa} \times T_{trav}$, where the drone remains at each trajectory point for ${\kappa}$ aggregation rounds. Please note that in this work, we have focused on achieving the best final performance optimizing the trajectory points for the asymptotic case where $T \to \infty$, and not on the training time minimization. However, as we show in our simulation results, our approach also significantly reduces the number of aggregation rounds $T$ required to achieve a target accuracy.\color{black}

\color{black}To get better insights on the factors impacting the drone's movement, let us assume $\kappa=1$, i.e. the drone stays at each trajectory point only for one aggregation round, and relax the $\boldsymbol{p}^1=\boldsymbol{p}^{T+1}$ constraint. An extension to $\kappa>1$ is straightforward. With these assumptions, the next optimal position for the drone at each aggregation round can be analytically derived in a sequential fashion. Consider the drone position at two consecutive aggregation rounds $\boldsymbol{p}^{t}=(x^{t-1}+v_x^t,y^{t-1}+v_y^t)$ and $\boldsymbol{p}^{t-1}=(x^{t-1},y^{t-1})$, it is straightforward to derive the following 

\begin{align}\label{Derivative4}
    \frac{\partial (J_t+K_t)}{\partial e_{i,t}}&=\frac{D_i}{D} (\frac{2c_1}{L}-\frac{\eta M \sigma_i^2}{2L D}), \\ \nonumber
    \frac{\partial e_{i,t}}{\partial d_{i,t}}& \approx \frac{\theta B N_0 \alpha (d_{i,t-1})^{\alpha-1}}{\rho_i A_i} \exp{\left(-\frac{\theta B N_0}{\mathbb{E}|h_{i,t}|^2 \rho_i}\right)},  \\ \nonumber\frac{\partial d_{i,t}}{\partial v_x^{t}} &\approx \frac{(x^{t-1}-x_i^{t-1})+(v_x^t-v^t_{x_i})}{d_{i,t-1}}=\frac{\Tilde{x}_i^{t-1}+(v_x^t-v^t_{x_i})}{d_{i,t-1}}, \\ \nonumber
    \frac{\partial d_{i,t}}{\partial v_y^{t}} &\approx \frac{(y^{t-1}-y_i^{t-1})+(v_y^t-v^t_{y_i})}{d_{i,t-1}}=\frac{\Tilde{y}_i^{t-1}+(v_y^t-v^t_{y_i})}{d_{i,t-1}},
\end{align}
where $\boldsymbol{v}^t_i=(v^t_{x_i}, v^t_{y_i})$ denotes the instantaneous velocity of the $i$'th device at aggregation round $t$. Referring to (\ref{Theorem}), to minimize the contribution of round $t$ to the ATL, we minimize $J_t+K_t$ by solving $\frac{\partial (J_t+K_t)}{\partial v_x^t}=0$ and $\frac{\partial (J_t+K_t)}{\partial v_y^t}=0$ and get the optimal instantaneous drone velocity ${\boldsymbol{\hat{v}}}^t=(\hat{v}_x^t, \hat{v}_y^t)$ as (\ref{Derivative5}). \color{black}According to (\ref{Derivative5}), the drone's optimal  instantaneous velocity is a weighted average of the devices' individual velocities. For the special case where different devices have similar dataset noise levels and channel conditions, the drone's optimal velocity reduces to the simple weighted average of the devices' individual velocities, weighted by the size of local datasets at devices. \color{black}Using (\ref{Derivative5}), we propose an alternative trajectory optimization technique based on the projected gradient method. In this technique, the drone starts from the location $({x^*}^{0},{y^*}^{0})$ optimized using our proposed \textbf{Algorithm 1} at $t=0$. Then, at each aggregation round $t$ we use (\ref{Derivative5}) to calculate $({{\hat{v}_x^t}}, {{\hat{v}_y^t}})$. Using projected gradient, we use ${{{{v^*}}_x^t}}=\min\{{{\hat{v}_x^t}}, v_{x, max}\}$ and ${{{{v^*}}_y^t}}=\min\{{{\hat{v}_y^t}}, v_{y, max}\}$ for the drone speed at aggregation round $t$ that satisfies the drone's maximum velocity constraint, i.e. $v^2_{max}=v^2_{x, max}+v^2_{y, max}$. 
Thereby we set $({x^*}^{t},{y^*}^{t})=({x^*}^{t-1}+{{{{v^*}}_x^t}},{y^*}^{t-1}+{{{{v^*}}_y^t}})$, and this process continues until the end of the training, i.e. for $T$ aggregation rounds. \color{black}\textbf{Algorithm 2} presents trajectory optimization with the proposed method.\color{black} 

Equations (\ref{Derivative5}) also explicitly drives effects of different learning (e.g., local dataset size and its noise level) as well as wireless factors (e.g., channel noise and fading, transmit power and bandwidth, etc.) on the optimal drone movement and its velocity. It should also be noted that\textbf{Algorithm 2} treats the aggregation rounds independently and thereby is suboptimum for trajectory optimization over the horizon of $T$ rounds. However, it is much simpler than the SQP-based method and the simulation results show sufficient quality of this solution leading to a significant performance improvement in comparison with other baselines.
\color{black}




\section{Simulation Results and Analysis}\label{sec:Results}
\textcolor{black}For our simulations, we consider a group of intelligent wireless devices collaborating to train a shared NN model with the help of a drone. We consider two simulation scenarios of collaborative image recognition for AIoT devices and semantic segmentation for autonomous vehicles. The general simulation parameters are listed in Table \ref{tbl:SisPars}. \color{black}We compare our results with two baselines. The first baseline is to position the drone at the geometrically weighted centroid of the devices for each aggregation round, where the weighting is selected proportional to the amount of data available to each user in order to position the drone closer to the device with more data. The second baseline is \cite{Rui} which designs the drone trajectory to maximize the average achievable communication rate, i.e. $\Sigma_{i=1}^{N}\log_{2}({1+\frac{{\mathbb{E}|h_{i}|^2 \rho_i}}{B{N_0}}})$. These two baselines are called respectively ``Weighted centroid" and ``Maximum rate \cite{Rui}".\color{black}

\color{black}

\begin{table}[t]
\centering
\caption{\color{black} Simulation parameters.}
\label{tbl:SisPars}
\begin{tabular}{|c|c|c|c|}
\hline
\color{black} Parameter & \color{black} Value & \color{black} Parameter & \color{black} Value \\ \hline
\color{black} $H$ & \color{black} $20m$ & \color{black} $N_0$ & \color{black} $-174$ dBm/Hz \\ \hline
\color{black} $\theta$ & \color{black} $0.053$ dB & \color{black} $\alpha$ & \color{black} $-3.4$ \\ \hline
\color{black} $f_c$ & \color{black} $1$ GHz & \color{black} $\nu_i$ & \color{black} $\in U[0.1,1]$ \\\hline
\end{tabular}
\end{table}

\subsection{Collaborative image recognition for AIoT devices}


\color{black}In this scenario, we consider $N=5$ AIoT devices located in a $70$ m $\times 70$ m square area. Considering the limited battery capacity of AIoT devices and the fact that only a small portion of the device transmit power and battery capacity can be dedicated to collaborative learning, we assume $\rho_i=0.1$ mW, $B=2.5$MHz. For image recognition, we use a lightweight two-layer fully connected NN for hand-written digit recognition using MNIST samples. The NN consists of a hidden dense layer with $200$ neurons and Rectified Linear Unit (ReLU) activation, followed by an output dense layer with $10$ neurons. The training samples are heterogeneously distributed, with devices having various numbers of training samples and different data qualities, i.e., different peak-signal-to-noise-ratio (PSNR) values. We consider both identical (iid) and non-identical (non-iid) class distributions, wherein in the iid case all classes constitute a roughly balanced portion of the local device datasets, while this portion is severely unbalanced for the non-iid case. The earth mover's distance (EMD) measures this class heterogeneity, and its value is $0.2728$ and $0.7624$ for the iid and non-iid cases, respectively \cite{zhao2018federated}. Other simulation parameters are $M=28\times28, c_1=1, c_2=0.5, \eta=0.8$. In the following, we investigate both stationary and moving device scenarios, as well as the sensor noise effects.\color{black}

\subsubsection{Drone placement for stationary devices:} \noindent In this scenario, we assume stationary devices located in the $70$ m $\times 70$ m area uniformly at random as demonstrated in Fig. ~\ref{fig:Env}. Fig. ~\ref{fig:Env} also shows the number of data samples at each device, i.e., $D_i$, and the local data qualities represented by the PSNR values, i.e., $\zeta_i$. For stationary devices and drone, we find the optimal drone location by solving optimization (\ref{eq:ATL_placement}) using the fractional programming technique proposed in Section \ref{placement}.



The training curves presented in Fig. ~\ref{fig:pl_curve} show that both the convergence speed and the final accuracy of the trained model considerably improve when the drone position is optimized by the proposed approach. The improvement in the final accuracy of the trained model in comparison with the ``Maximum  rate" and ``Weighted centroid" baselines is roughly $4\%$ and $2\%$ for the iid case, and $4.5\%$ and $4.2\%$ for the non-iid case. To achieve a target accuracy of $75\%$ in the iid case, our proposed approach requires only $85$ aggregation rounds, whereas the ``Maximum rate" and ``Weighted centroid" baselines require $96$ and $104$ aggregation rounds, demonstrating $11.5\%$ and $18\%$ reduction, respectively. For the non-iid case, $75\%$ accuracy is achieved after only $77$ aggregation rounds with our proposed approach, but requires roughly $108$ and $123$ aggregations for the ``Maximum  rate" and ``Weighted centroid" baselines, representing $28.7\%$ and $37\%$ reduction, respectively. 

\begin{strip}
\rule{\textwidth}{0.75pt} 
\begin{align}\label{Derivative5}
    {\mathlarger{\hat{v}_x^t}} & {\mathlarger{=\frac{\sum_{i=1}^{N} \frac{D_i}{D} (\frac{2c_1}{L}-\frac{\eta M \sigma_i^2}{2L D}) (\frac{\theta B N_0 \alpha (d_{i,t-1})^{\alpha-2}}{\rho_i A_i} ) \exp{\left(-\frac{\theta B N_0}{\mathbb{E}|h_{i,t}|^2 \rho_i}\right)} (v^t_{x_i}-\Tilde{x}_i^{t-1})}{\sum_{i=1}^{N} \frac{D_i}{D} (\frac{2c_1}{L}-\frac{\eta M \sigma_i^2}{2L D}) (\frac{\theta B N_0 \alpha (d_{i,t-1})^{\alpha-2}}{\rho_i A_i} ) \exp{\left(-\frac{\theta B N_0}{\mathbb{E}|h_{i,t}|^2 \rho_i}\right)}},}} \\ \nonumber
    {\mathlarger{\hat{v}_y^t}}& {\mathlarger{=\frac{\sum_{i=1}^{N} \frac{D_i}{D} (\frac{2c_1}{L}-\frac{\eta M \sigma_i^2}{2L D}) (\frac{\theta B N_0 \alpha (d_{i,t-1})^{\alpha-2}}{\rho_i A_i} ) \exp{\left(-\frac{\theta B N_0}{\mathbb{E}|h_{i,t}|^2 \rho_i}\right)} (v^t_{y_i}-\Tilde{y}_i^{t-1})}{\sum_{i=1}^{N} \frac{D_i}{D} (\frac{2c_1}{L}-\frac{\eta M \sigma_i^2}{2L D}) (\frac{\theta B N_0 \alpha (d_{i,t-1})^{\alpha-2}}{\rho_i A_i} ) \exp{\left(-\frac{\theta B N_0}{\mathbb{E}|h_{i,t}|^2 \rho_i}\right)}}.}}
\end{align}

\end{strip}

\begin{figure}[t]
    \centering    \includegraphics[scale=.3]{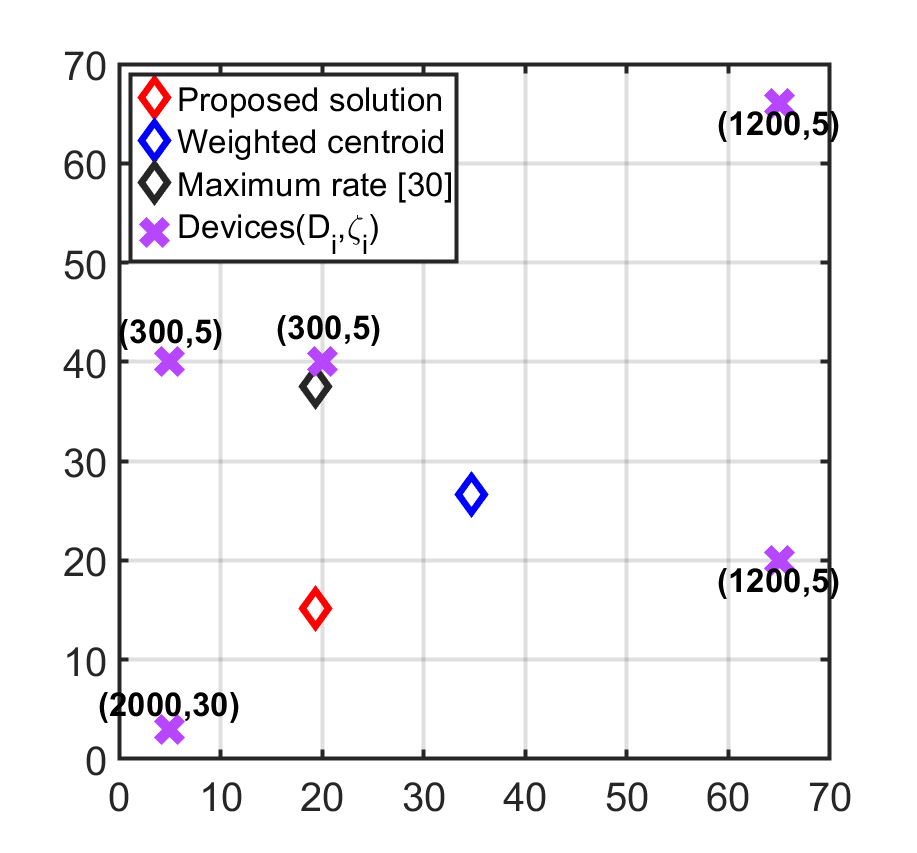}
    \caption{Stationary AIoT devices scenario.}
    \label{fig:Env}
\end{figure}

\color{black}It should be noted that handwritten digit recognition on MNIST is a simple task with a typical accuracy of above $99\%$ achievable in the ideal centralised learning scenario when there are no packet errors, the dataset is noiseless, and the model is trained on the whole 50000 training data in MNIST dataset. However, in our AIoT scenario due to: 1-The limited storage and computation capacity at the AIoT devices allowing much fewer data, i.e. $\sum_{i=1}^{N} D_i=5000$, for training; 2-The noisy datasets with PSNRs as low as 5dB; 3-Presence of the packet errors; 4-Heterogenious data with non-zero EMD causing weight divergence [31] in comparison with centralized learning; we get a lower final accuracy, i.e. $\sim 80\%$ in these curves.\color{black}



\begin{figure}[t]
    \centering    \includegraphics[scale=.25]{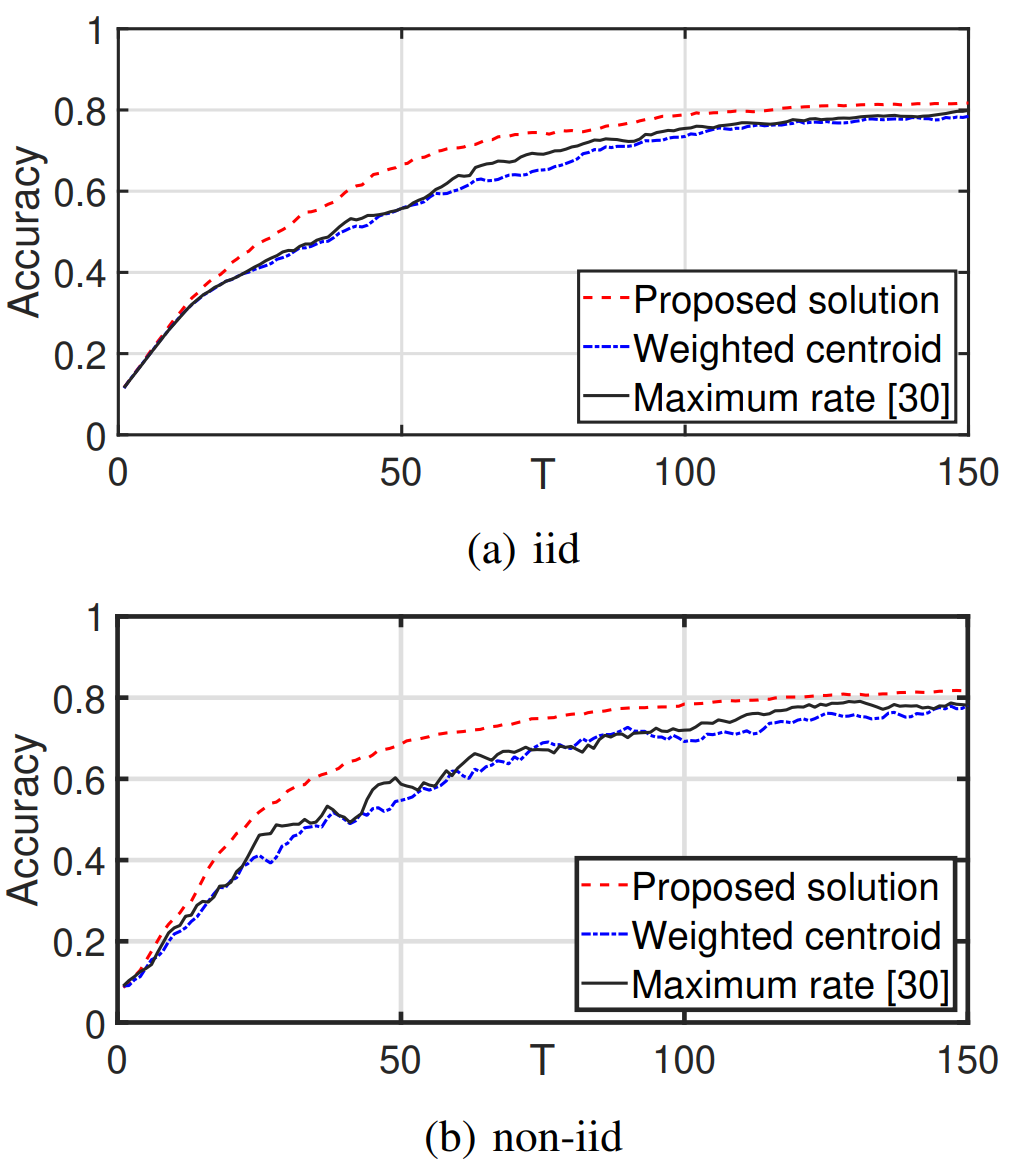}
    \caption{Training curves for drone-assisted collaborative learning on stationary AIoT devices.}
    \label{fig:psnr_curve}
\end{figure}

\subsubsection{Trajectory optimization for moving devices:}
\noindent In this scenario, we consider moving devices, and determine the optimal trajectory points by solving optimization (\ref{eq:ATL_Optimization}) using SQP as proposed in section \ref{trajectory}. We pick the initial device locations in the area uniformly at random and assign random speeds to each device according to $v_{x_i}, v_{y_i} = 0.1 \times U[0,1]$, where $U[.,.]$ is the uniform distribution. The start and end positions of the devices and the drone is shown in Fig. \ref{fig:trajectory_presentation}. \color{black}As shown in Fig. 2 and Fig. 5, the optimal placement and trajectory for the drone is closer to device 5 which has a larger and less noisy dataset with a larger PSNR of $\zeta_5=30$dB. It is clear in these figures that the other two baselines (i.e. Weighted centroid, and Maximum rate [30]) that do not consider the sensor noise effects result in totally different placements and trajectories which are suboptimum and result in a lower final accuracy for the trained model as shown in Fig. ~\ref{fig:pl_curve} and Fig. ~\ref{fig:trajectory_optimization}.\color{black}

\color{black}
\begin{figure}[t]
    \centering    \includegraphics[scale=.25]{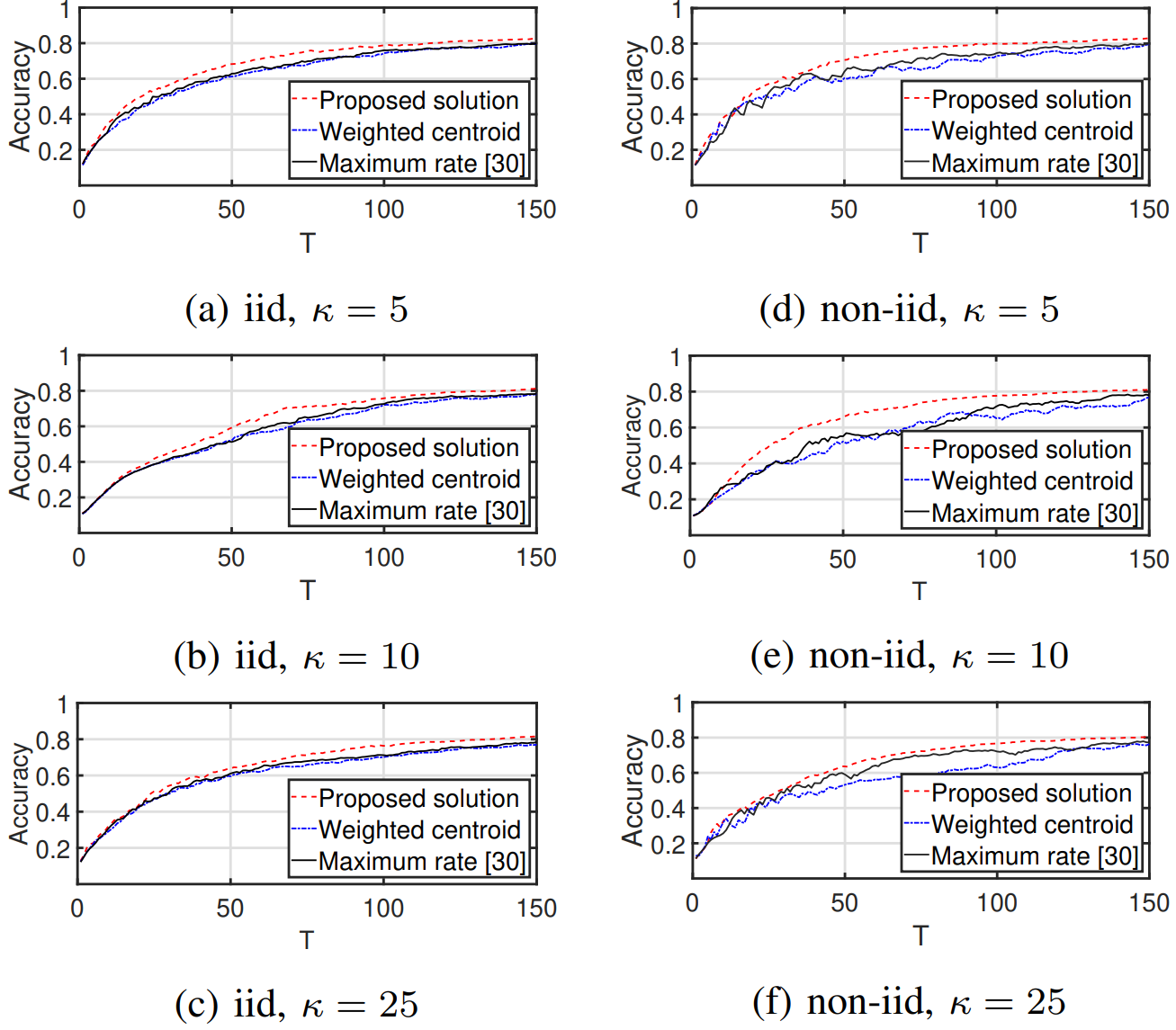}
    \caption{Learning curves for moving AIoT devices scenario.}
    \label{fig:psnr_curve}
\end{figure}

The training curves presented in Fig. ~\ref{fig:trajectory_optimization} demonstrate that both the convergence speed and the final accuracy of the trained model considerably improve when the drone trajectory is optimized by the proposed approach. The improvement in the final accuracy of the trained model in comparison with the ``Maximum  rate" and ``Weighted centroid" baselines is given in Table ~\ref{tab:trajectory_imp}. The final model accuracy increases with smaller $k$ for both iid and non-iid cases. In the iid case with $\kappa=5$, a target accuracy around $75\%$ is achieved after $65$ aggregation rounds with the proposed approach, whereas this roughly requires $85$ aggregation rounds for the two other baselines, demonstrating $18.7\%$ improvement in terms of the speed of convergence. Similarly in the non-iid case with $\kappa=5$, a target accuracy of around $75\%$ is reached after $65$ aggregation rounds, while this is around $82$ and $99$ aggregation rounds for the ``Maximum  rate" and ``Weighted centroid" respectively, representing approximately $21\%$ and $34\%$ improvement in terms of the convergence speed.

\begin{figure}[t]
    \centering    \includegraphics[scale=.6]{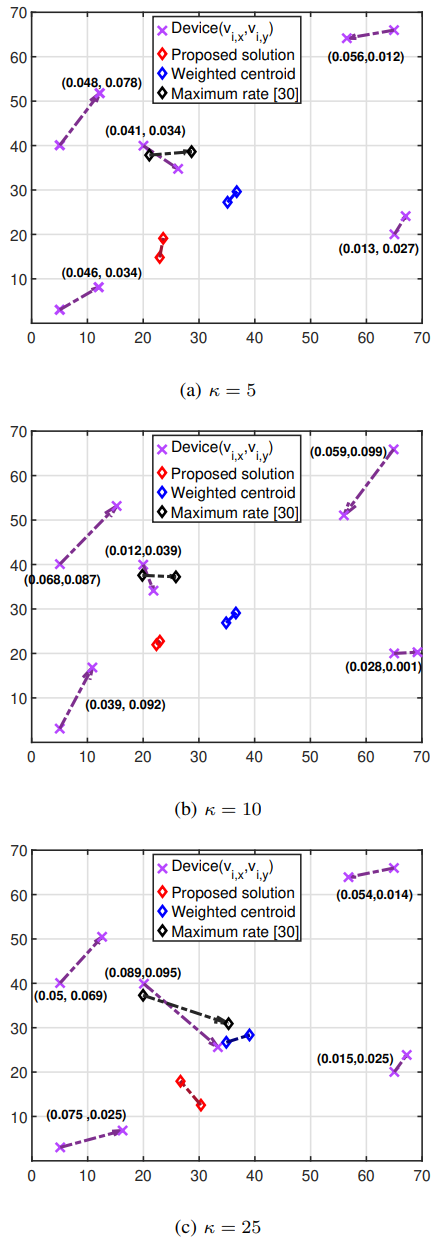}
    \caption{Trajectory optimization for collaborative drone-assisted image recognition on moving AIoT devices. The start and end positions of mobile devices and the drone for $k=5, 10, 25$.}
    \label{fig:trajectory_presentation}
\end{figure}

\begin{table}[t]
\centering
\caption{\color{black}Comparisons of the final accuracy with baselines for drone-assisted collaborative image recognition on AIoT devices.}
\label{tab:trajectory_imp}
\resizebox{\columnwidth}{!}{%
\begin{tabular}{|ll|l|l|l|}
\hline
\multicolumn{2}{|l|}{} & \color{black} $k=5$ & \color{black} $k=10$ & \color{black} $k=25$ \\ \hline
\multicolumn{1}{|c|}{\multirow{2}{*}{\color{black} iid}} & \color{black} Proposed & \color{black} 83.4\% & \color{black} 82.1\% & \color{black} 80.7\% \\ \cline{2-5}
\multicolumn{1}{|c|}{} & \color{black} Weighted centroid & \color{black} 78.8\% & \color{black} 78.1\% & \color{black} 76.9\% \\ \cline{2-5}
\multicolumn{1}{|c|}{} & \color{black} Maximum  rate \cite{Rui} & \color{black} 79.6\% & \color{black} 78.4\% & \color{black} 77.6\% \\ \hline
\multicolumn{1}{|c|}{\multirow{2}{*}{\color{black} Non-iid}} & \color{black} Proposed & \color{black} 82.7\% & \color{black} 81.3\% & \color{black} 80.1\% \\ \cline{2-5}
\multicolumn{1}{|c|}{} & \color{black} Weighted centroid & \color{black} 77.9\% & \color{black} 77.1\% & \color{black} 76.0\% \\ \cline{2-5}
\multicolumn{1}{|c|}{} & \color{black} Maximum  rate \cite{Rui} & \color{black} 78.5\% & \color{black} 77.7\% & \color{black} 77.0\% \\ \hline
\end{tabular}%
}
\end{table}

\color{black}Let us also compare the results in Table~\ref{tab:trajectory_imp} with the average accuracy achievable when we use a fixed (i.e., non-mobile) model aggregator located at a random
position in the $70m \times 70m$ area at the same altitude $H=20m$. This accuracy is $75.1$ and $73.6$ for the iid and non-iid cases, respectively. The improvements we achieve in Table~\ref{tab:trajectory_imp} in comparison with these numbers shows the benefits of mobility of the aggregator and is due to our proposed trajectory optimization algorithm that is aware of both the dataset size and qualities as well as the channel conditions. As shown in Fig. 2, if a device has a larger dataset of clean data, with our proposed algorithm, the drone will hover closer to that device to reduce the corresponding packet errors, thereby improving the learning performance.\color{black}

\color{black}
\subsubsection{Packet error rate effects:}
In this subsection, we investigate effects of the packet errors on the convergence speed and final accuracy of collaborative learning. We study effects of the packet errors on the ATL in Fig. \ref{PER}, where the drone and the devices are fixed, the sensor noise levels are $\zeta_1=\hdots=\zeta_5=5$ dB, and all the other simulation parameters are the same as in Subsection IV.A.1. We plot the ATL for various $D_5/D$ values where the packet error rate for device 5, i.e. $e_5$, spans the large range of values $0.0001$ to $0.01$. As the packet error rate increases, the ATL increases thereby degrading the final accuracy of the trained NN.

To study effects of the packet errors on the convergence speed and final accuracy, we train our classification NN for the iid case where the drone/devices are stationary and the packet error rate for device 5 gets two different values, i.e. $e_5=0.01, 0.1$. Fig. \ref{PER1} provides the resulting learning curves. We get a $2.7$\% improvement in the final accuracy of the trained model as $e_5$ decreases from $0.1$ to $0.01$. Moreover, to achieve a target accuracy of $75\%$, we need 113 aggregation rounds when $e_5=0.1$, while we need 81 aggregation rounds when $e_5=0.01$ which shows slower convergence as the packet error rate increases. This figure and the corresponding paragraph is added to Subsection IV.A.3 of the revised manuscript and provided here for easier reference.\color{black}

\begin{figure}[t]
         \centering         
         \includegraphics[scale=.27]{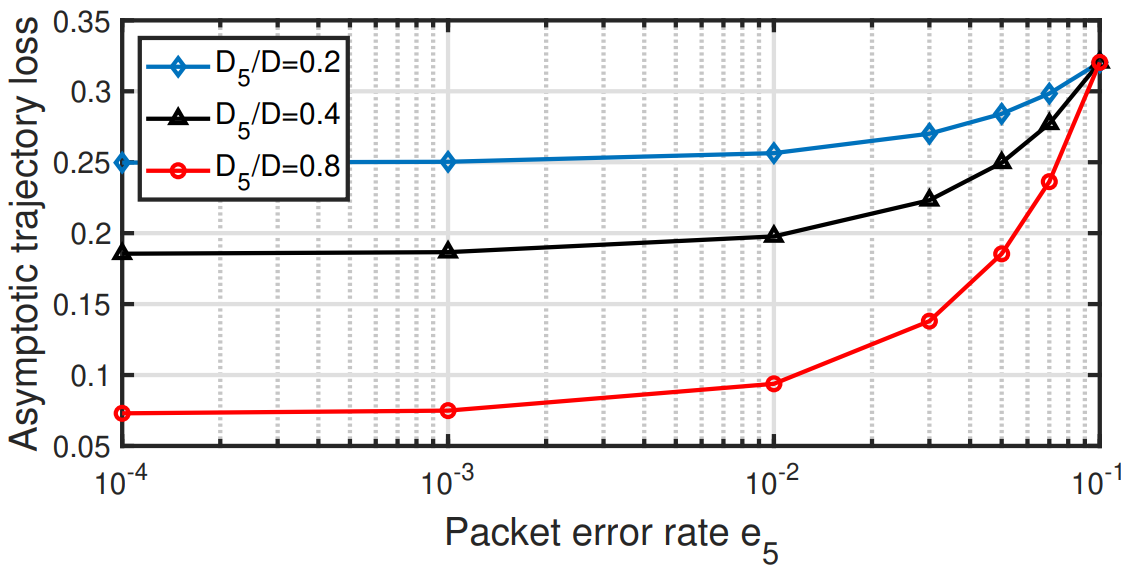}
         \caption{\color{black}Effects of the packet error rate on the asymptotic trajectory loss.}
         \label{PER}
    \end{figure}

    \begin{figure}[t]
         \centering         
         \includegraphics[scale=.25]{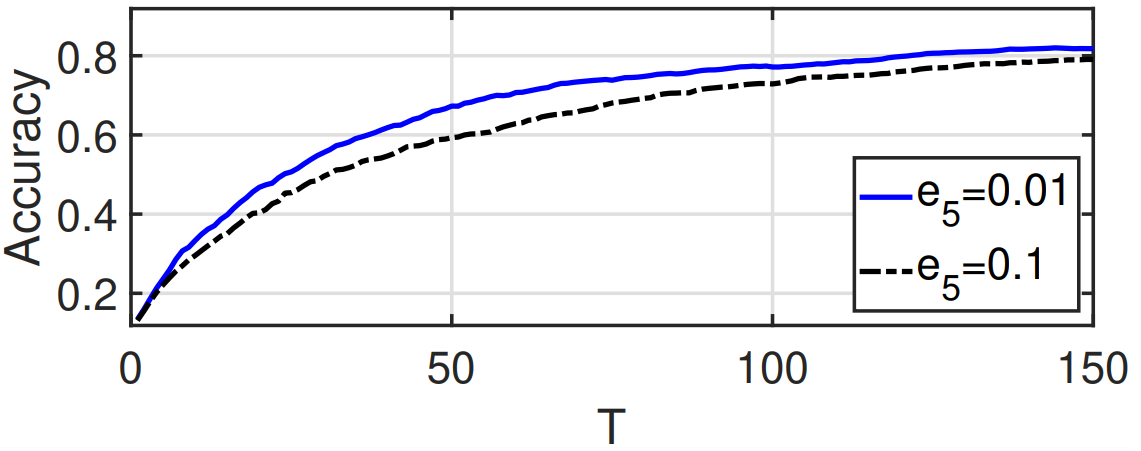}
         \caption{\color{black}Effects of the packet error rate on the learning curves.}
         \label{PER1}
    \end{figure}

\subsubsection{Sensor noise effects:}
Next, we investigate the impact of the sensor noise levels on the convergence speed and final performance of the trained NN. \color{black}Fig. \ref{PSNR} studies effects of the dataset noise and its size on the asymptotic trajectory loss. This figure plots the ATL versus the PSNR of device 5 $\zeta_5$, when $\zeta_1=\hdots=\zeta_4=5$ dB for three cases where device 5 has 20\%, 40\%, and 80\% of the total data. As $\zeta_5$ increases, the ATL decreases thereby leading to a better final performance of the trained model. The changes are more significant for $\frac{D_5}{D}=0.8$.\color{black}

\noindent Next, we investigate the impact of the sensor noise levels on the convergence speed and final performance of the trained NN. We plot the learning curves for a trajectory optimization scenario in Fig. ~\ref{fig:psnr_curve}, where the number of trajectory points $k=10$, the PSNR values $\{\zeta_i\}_{i=1}^{4}=5 \mathrm{dB}$ and is kept fixed while $\zeta_5 \in \{0, 5, 30\} \mathrm{dB}$. As Fig. ~\ref{fig:psnr_curve} shows, the sensor noise level considerably affects the final accuracy of the trained NN, but does not affect the convergence speed much, which confirms our convergence analysis in the Theorem \ref{thm:Thrm}.

 \begin{figure}[t]
         \centering         
         \includegraphics[scale=.25]{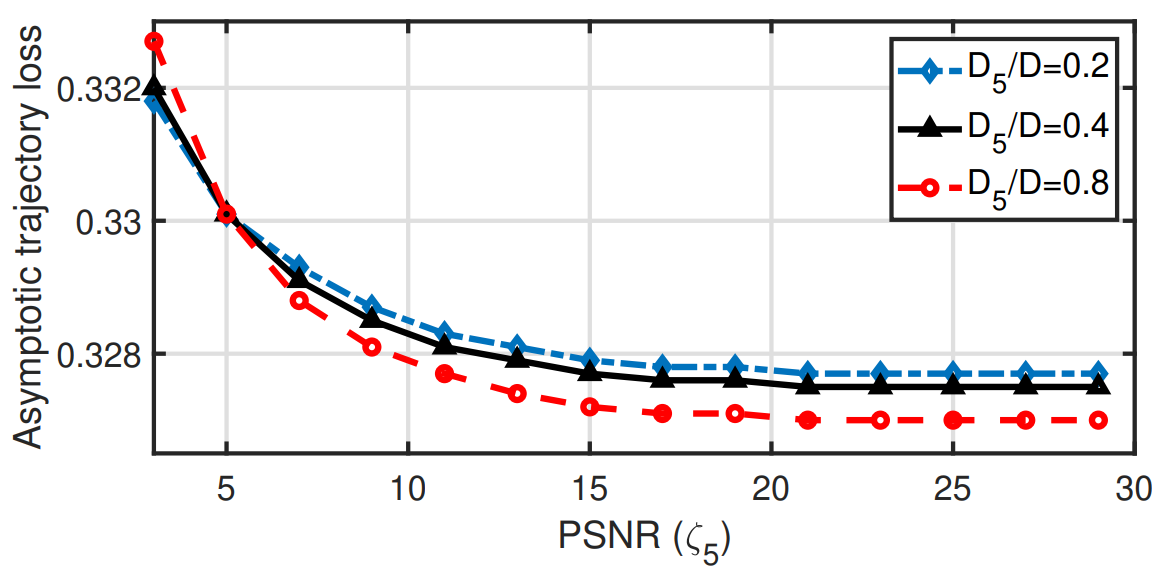}
         \caption{\color{black}Effects of the local dataset noise on asymptotic trajectory loss.}
         \label{PSNR}
\end{figure}
\color{black}

\begin{figure}[t]
    \centering    \includegraphics[scale=.65]{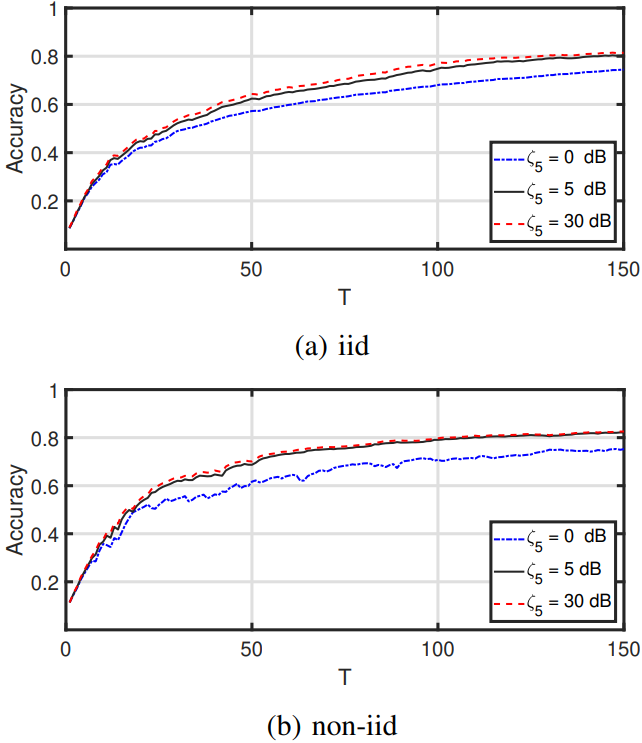}
    \caption{Learning curves for different PSNR values on moving AIoT devices.}
    \label{fig:psnr_curve}
\end{figure}


\begin{table*}[t]
\centering
\caption{\color{black}Classification accuracy for drone-assisted collaborative image recognition on AIoT devices.}
\label{tbl:WalidCompare}
\begin{tabular}{|ll|l|l|l|l|l|l|l|}
\hline
\multicolumn{2}{|l|}{} & $T=1$ & $T=15$ & $T=30$ &  $T=60$ & $T=90$ &  $T=120$ &  $T=150$ \\ \hline
\multicolumn{1}{|c|}{\multirow{2}{*}{iid}} & Proposed & 11.60\% & 41.14\% & 54.49\% & 66.97\% &  74.91\% &  78.33\% &  80.85\% \\ \cline{2-9} 
\multicolumn{1}{|c|}{} & Noise unaware \cite{Joint} & 11.51\% & 39.66\% & 52.14\% &  64.15\% &  71.91\% &  77.51\% & 79.73\% \\ \hline
\multicolumn{1}{|l|}{\multirow{2}{*}{Non-iid}} & Proposed & 10.84\% & 38.20\% & 51.98\% & 65.32\% &  72.75\% &  77.27\%  & 80.20\% \\ \cline{2-9} 
\multicolumn{1}{|l|}{} & Noise unaware \cite{Joint} & 10.83\% & 36.38\% & 51.39\% &  63.99\%  & 71.54\%  & 76.13\%  & 78.81\% \\ \hline
\end{tabular}%
\end{table*}

We also compare the NN accuracy when the drone trajectory is optimized using our noise-aware approach (\ref{eq:ATL_Optimization}) versus \cite{Joint}, in Table. ~\ref{tbl:WalidCompare} for $k=25$. This Table shows a considerable improvement by our noise-aware approach in the accuracy of the trained NN for various number of aggregation rounds for both the iid and non-iid class data distributions. For $T=150$, the improvement is 1.4 \% and 1.8 \% for the iid and non-iid cases, respectively. \color{black}In terms of the convergence speed, these results demonstrate an   accuracy of $64\%$ achievable with roughly $50$ aggregation rounds in our proposed solution, while the same accuracy would require at least $60$  aggregation rounds for \cite{Joint}, demonstrating $~17\%$ improvement in the drone's hovering time, battery usage, and communication overhead for collaborative learning. \color{black}

\color{black}

\color{black}
\subsubsection{Drone altitude and speed effects:}
In this subsection, we study relations between the drone flight characteristics, i.e. altitude and maximum speed, on the learning performance. To this end, we evaluate the Asymptotic Trajectory Loss (ATL) (as defined in III.B), which is a surrogate indicator of the corresponding final performance of the trained NN when the drone takes a specific trajectory during training. In these simulations, the data is iid and all other simulation parameters stay the same as IV.A.2. 

\color{black}Fig. \ref{height} provides the ATL for various drone altitudes. In Fig. \ref{height}, we have considered four Values for the drone altitude, $H={5, 10, 15, 20}m$ and assume $\zeta_{i,t} \approx 1$ for all cases. As $H$ increases, the average channel gain $\mathbb{E}|h_{i,t}|^2$ decreases, thereby increasing the packet error rates $e_{i,t}$. Referring to \textbf{Theorem 1}, the increased packet error rate increases the ATL, thereby reducing the final accuracy of the trained NN. Fig. \ref{height} demonstrates a consistent reduction in the ATL values by our proposed method in comparison with the benchmarks. The average ATL improvement by our proposed method is $47\%$ and $40\%$ in comparison with the ``Weighted centroid" and ``Maximum rate [30]" benchmarks, respectively.\color{black}

 \begin{figure}[t]
         \centering         
         \includegraphics[scale=.25]{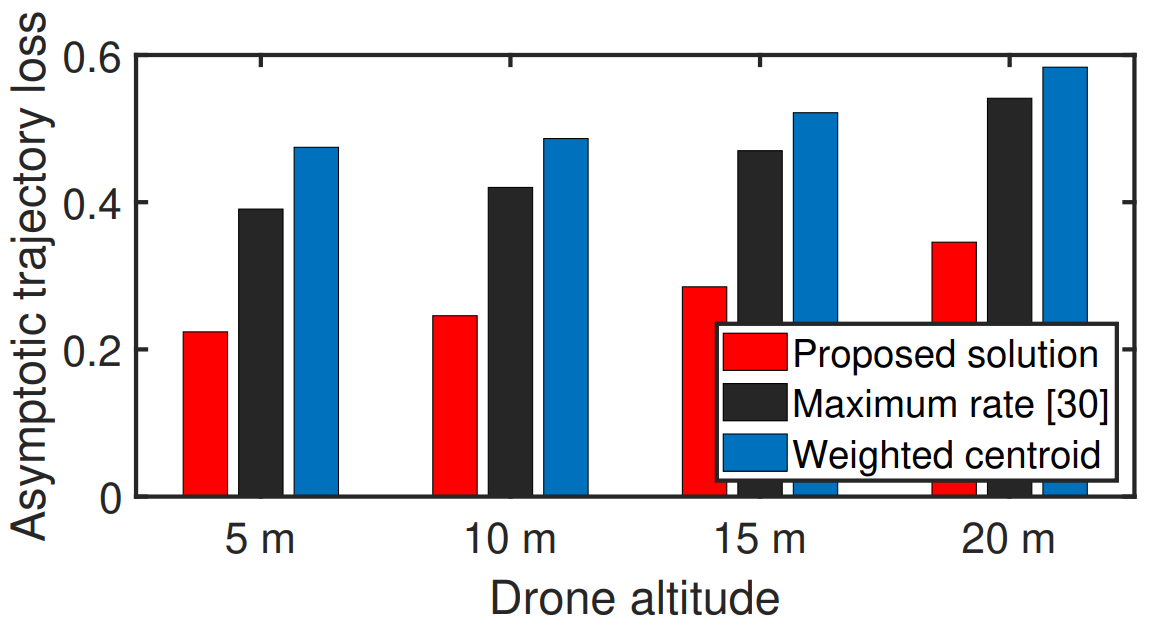}
         \caption{\color{black}The asymptotic trajectory loss for various drone altitudes.}
         \label{height}
\end{figure}
\color{black}

\color{black}

\color{black}Fig. \ref{vmax} plots ATL versus $v_{max}$ for two different altitudes $H=5, 20$m. As $v_{max}$ increases, the ATL decreases leading to a better performance for the trained NN. However, for larger $v_{max}$ values, i.e. $v_{max} \le 4 $, the constraint corresponding to the maximum drone speed becomes inactive and the ATL approaches a constant value.\color{black}

\begin{figure}[t]
         \centering         
         \includegraphics[scale=.25]{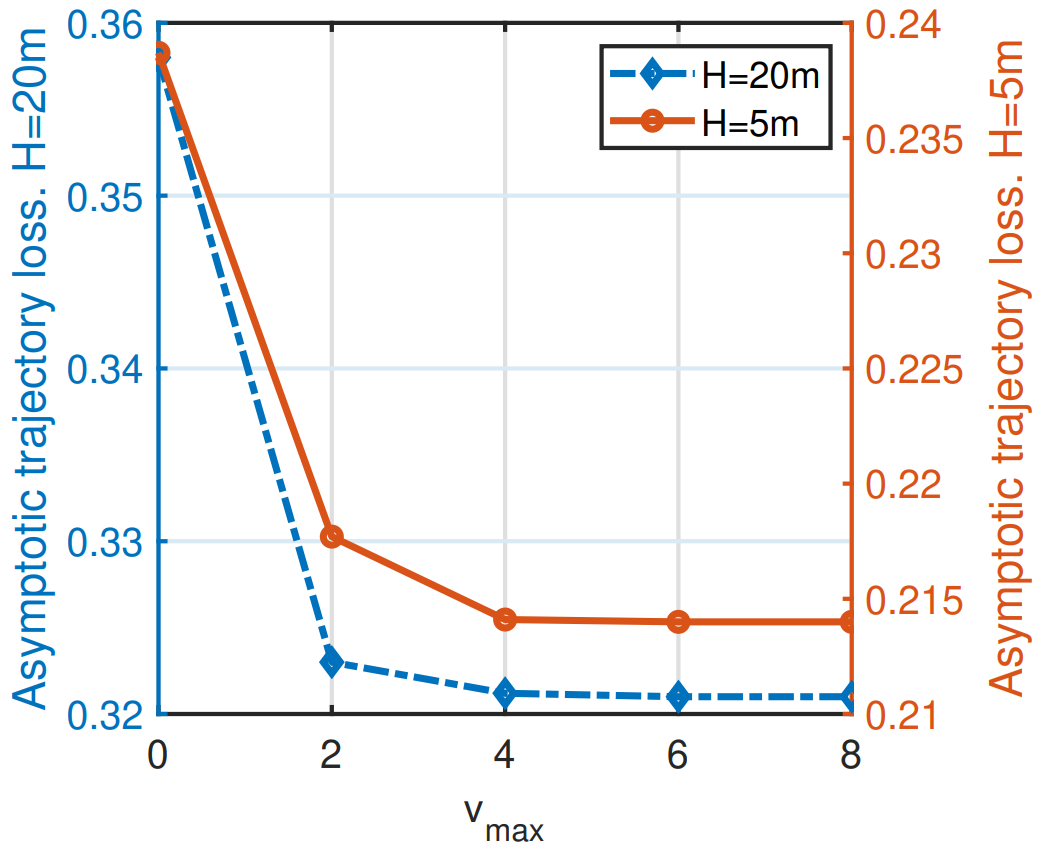}
         \caption{\color{black}Effects of $v_{max}$ on the asymptotic trajectory loss.}
         \label{vmax}
\end{figure}


\color{black}


\begin{figure*}[t]
     \centering
     \includegraphics[scale=0.55]{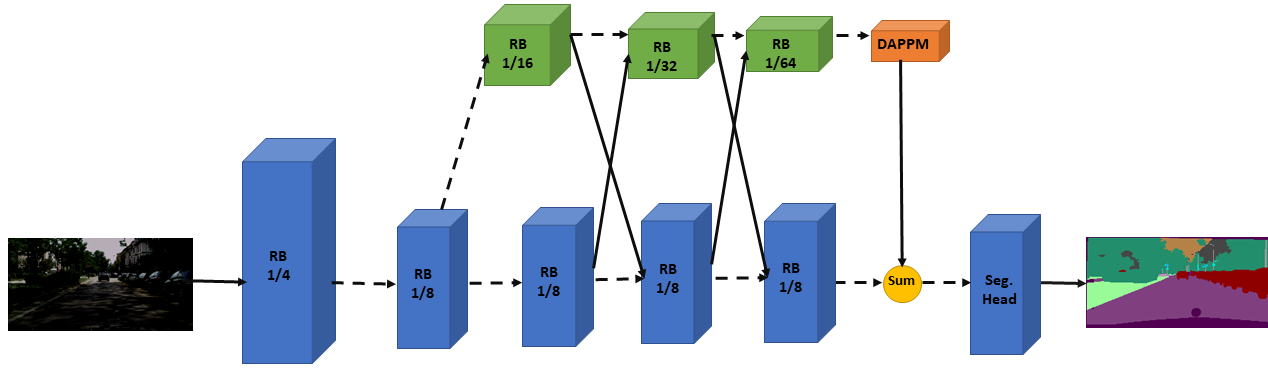}
     \caption{\color{black}The overall architecture of DDRNets \cite{hong2021deep} for image semantic segmentation.}
     \label{fig:DDRNets}
\end{figure*}

\subsection{Collaborative semantic segmentation for autonomous vehicles}
Semantic segmentation is an essential learning task for autonomous vehicles to gain an accurate understand of their surrounding area. In this section, we evaluate performance of our proposed trajectory optimization for collaborative training of a DNN-based semantic segmentation model on autonomous vehicles. \color{black}In this scenario, we have $N=5$ moving vehicles located in a $350$ m $\times 350$ m square area. The velocities are assigned randomly according to $v_{x_i}, v_{y_i} = U[0,13]$, where $U[.,.]$ is the uniform distribution, and we assume $\rho_i=8$ mW, $B=1.5$ MHz, $\kappa=10$. The amount of data at each vehicle and the PSNR values are $D_i=\{1190, 178, 179, 714, 714\}$  and $\zeta_i=\{30,5,5,5,5\}$, respectively. For performance comparison, we use both pixel accuracy (PA) and intersection-over-union (IoU) metrics. PA captures the percentage of pixels that are accurately classified in the image. IoU is a common metric for object detection and semantic image segmentation tasks that quantifies the degree of overlap between the predicted and ground truth object class masks and can better indicate accuracy in the placement of the predicted objects. For object class $i$, the IoU metric is defined as $\mathrm{IoU} (i) = \frac{A_o(i)}{A_u(i)}$ where $A_o(i)$ and $A_u(i)$ are the areas of overlap and union between the ground truth and predicted masks for class $i$, respectively.\color{black}

We train our model on Cityscapes dataset which is a benchmark for semantic understanding of urban street scenes \cite{Cordts2016Cityscapes}. This dataset consists of $2975$ and $500$ sample street scenes for training and testing, respectively. The dataset includes $8$ different semantic categories, namely, construction, object, nature, sky, human, vehicle, flat, and void, and a total of $20$ object classes. Each category may include several object classes, e.g., the vehicle category covers $``Bicycle"$, $``Motorcycle"$, $``Train"$, $``Trailer"$, $``Caravan"$, and $``Car"$ object classes.

We adopt deep dual-resolution networks (DDRNets) that were recently proposed for real-time semantic segmentation \cite{hong2021deep}. Deep dual-resolution networks (DDRNets) are composed of two deep branches, and a deep information extractor called Deep aggregation pyramid pooling module (DAPPM) to enlarge effective receptive fields and fuse multi-scale context based on low-resolution feature maps as shown in Fig. \ref{fig:DDRNets}. RB and RBB represent sequential and single residual bottleneck blocks. Information is extracted from low-resolution feature maps e.g., $1/8$ denotes high-resolution branch create feature maps whose resolution is 1/8 of the input image resolution. The black solid lines demonstrate the data processing path (including up-sampling and down-sampling) and black dashed lines denote information paths without data processing. $``Sum"$ denotes a
point-wise summation \cite{hong2021deep}.

\color{black}Table \ref{tab:SemPA} provides the values of average PA along the training with different trajectory optimization methods. Our proposed approach not only speeds up the training, but also achieves a better final performance. The maximum PA value achievable by centralized training using the Adam optimizer in this case is 71.35\% which is very close to the performance of our proposed solution. Fig. \ref{fig:cityscape_validation} compares the visual quality of segmentation between our proposed approach and the three other baselines for two sample street scenes in the dataset. Figs. \ref{fig:1_usg} and \ref{fig:2_usg} show the original scenes, while Figs. \ref{fig:1_sg} and \ref{fig:2_sg} show the ground truth segmented scenes. 
\cref{fig:op_1_gt,fig:op_2_gt,fig:op_comm_1_gt,fig:op_comm_2_gt,fig:wc_1_gt,fig:wc_2_gt,fig:walid_1_gt,fig:walid_2_gt} present the corresponding segmentation outputs when the drone trajectory is optimized using our proposed approach and the three other baselines. From this figure, we observe that our proposed noise-aware trajectory optimization approach provides more accurate segmentations. For example in the first scene, the proposed approach better detects the \color{black}$``Sidewalk"$\color{black} on the left while the ``Maximum  rate" and ``Weighted centroid" baselines fail to detect the  \color{black}$``Sidewalk"$\color{black} at all. Similarly, in the second scene, our proposed approach better detects the \color{black}$``Terrain"$\color{black} on the left and better locates the $``Car"$s in comparison with the other baselines. The approach based on \cite{Joint} mistakenly detects $``Sidewalk"$ on the right. More important is the accurate detection and placement of cars to minimize the risk of accidents and injuries. To numerically compare the detection and placement accuracy of cars, we report the IoU value for the target class $``Car"$ as well as the general PA metric below each image. Our proposed approach outperforms other baselines not only in terms of the PA, but also more importantly in terms of the IoU for the target class \color{black}$``Car"$\color{black}. 

\begin{table*}[t]
            \centering
            \caption{\color{black}PA for collaborative drone-assisted semantic segmentation with different trajectory optimization methods.}
            \label{tab:SemPA}
            \begin{tabular}{|l|l|l|l|l|l|l|l|l|l|}
            \hline
                                   & $T=1$    & $T=50$     & $T=150$    & $T=250$    & 
                                   $T=350$    &
                                   $T=450$    \\ \hline
            Proposed solution      & 22.34\% & 68.38\% & 68.02\% & 69.55\% & 69.94\% & 71.12\% \\ \hline
            Weighted centroid & 15.56\% & 64.33\% & 65.47\% & 66.52\% & 66.98\% & 67.58\% \\ \hline
            Maximum  rate \cite{Rui}           & 16.25\% & 61.94\% & 66.64\% & 66.68\% & 65.77\% & 67.53\%  \\ \hline
            {Noise unaware \cite{Joint}}               & 14.55\% & 66.72\% &66.77\% & 68.37\% & 69.05\% & 69.68\%  \\ \hline
            
            \end{tabular}
            \end{table*}



\section{Conclusions}\label{sec:Conclusions}
In this paper, we have considered collaborative training of a shared NN model by intelligent devices and orchestrated by a drone. We have then derived the convergence rate and final performance of the learning algorithm considering device data heterogeneity, sensor noise levels, and communication errors. Using this analysis, we have defined the ATL concept and, then, used it to find the drone trajectory that maximizes the final accuracy of the trained NN. The results show that our proposed trajectory optimization approach significantly improves the convergence rate and final performance in comparison with baselines that only consider device data characteristics  or wireless channel conditions and is specifically effective in saving the drone's hovering time, communication overhead, and battery usage.


\begin{appendices}\label{app:appendices}
        \centering {APPENDIX.}
        \begin{proof}[Proof of Theorem \ref{thm:Thrm}]\label{app:xx}
            \begin{FlushLeft}
            With assumption I, let us express $F(\boldsymbol{w}_{t+1})$ using its second-order Taylor expansion as 
            \begin{align}
                F(\boldsymbol{w}_{t+1})&=F(\boldsymbol{w}_t)+(\boldsymbol{w}_{t+1}-\boldsymbol{w}_t)^T \nabla F(\boldsymbol{w}_t)\\ \nonumber
                &+\frac{1}{2} (\boldsymbol{w}_{t+1}-\boldsymbol{w}_t)^T \nabla^2 F(\boldsymbol{w}_t) (\boldsymbol{w}_{t+1}-\boldsymbol{w}_t),\\ \nonumber
                & \le F(\boldsymbol{w}_t)+(\boldsymbol{w}_{t+1}-\boldsymbol{w}_t)^T \nabla F(\boldsymbol{w}_t)+\frac{L}{2}\|\boldsymbol{w}_{t+1}-\boldsymbol{w}_t\|^2,
            \end{align}
            where the inequality stems from assumption III. Given the learning rate $\lambda=\frac{1}{L}$, the expected loss function $\mathbb{E}(F(\boldsymbol{w}_{t+1}))$ can be expressed as
            \begin{align}\label{mainbound}
                \mathbb{E}(F(\boldsymbol{w}_{t+1})) &\le \mathbb{E} (F(\boldsymbol{w}_t)-\lambda(\nabla F(\boldsymbol{w}_t)-\boldsymbol{J})^T\nabla F(\boldsymbol{w}_t)\\ \nonumber
                &+\frac{L \lambda^2}{2}\|\nabla F(\boldsymbol{w}_t)-\boldsymbol{J}\|^2),\\ \nonumber
                &=\mathbb{E}(F(\boldsymbol{w}_t))-\frac{1}{2L}\|\nabla F(\boldsymbol{w}_t)\|^2+\frac{1}{2L}\mathbb{E}(\|\boldsymbol{J}\|^2),
            \end{align}
            where $\boldsymbol{J}$ is the gradient error due to sensor noise and packet errors given by
            \begin{align}\label{mainbound}
                \boldsymbol{J}=\nabla F(\boldsymbol{w}_t)-\frac{\sum_{i=1}^{N} \sum_{m=1}^{D_i} \nabla f(\boldsymbol{w}, \boldsymbol{x}_{im}+\boldsymbol{n}_{im}, y_{im})C(\boldsymbol{w}_{i,t})}{\sum_{i=1}^{N}D_i C(\boldsymbol{w}_{i,t})}.
            \end{align}
            With the convexity assumption we have $\nabla_{\boldsymbol{w}} f(\boldsymbol{w}_t, \boldsymbol{x}_{im} + \boldsymbol{n}_{im}, y_{im}) \ge \nabla_{\boldsymbol{w}} f(\boldsymbol{w}_t, \boldsymbol{x}_{im}, y_{im}) +\nabla_{\boldsymbol{wx}}^2 f(\boldsymbol{w}_t, \boldsymbol{x}_{im}, y_{im}) \boldsymbol{n}_{im}$, thereby we get
            \begin{align}\label{mainbound}
                \boldsymbol{J} &\le \left [\nabla F(\boldsymbol{w}_t)-\frac{\sum_{i=1}^{N} C(\boldsymbol{w}_{i,t}) \sum_{m=1}^{D_i} \nabla f(\boldsymbol{w}, \boldsymbol{x}_{im}, y_{im})}{\sum_{i=1}^{N}D_i C(\boldsymbol{w}_{i,t})} \right ] \\ \nonumber
                &- \left [\frac{\sum_{i=1}^{N} C(\boldsymbol{w}_{i,t}) \sum_{m=1}^{D_i} \nabla_{\boldsymbol{wx}}^2 f(\boldsymbol{w}_t, \boldsymbol{x}_{im}, y_{im}) \boldsymbol{n}_{im}}{D} \right ].
            \end{align}
            The first error term in (\ref{mainbound}), i.e.,
            \begin{align}
                \boldsymbol{A}=\nabla F(\boldsymbol{w}_t)-\frac{\sum_{i=1}^{N} C(\boldsymbol{w}_{i,t}) \sum_{m=1}^{D_i} \nabla f(\boldsymbol{w}, \boldsymbol{x}_{im}, y_{im})}{\sum_{i=1}^{N}D_i C(\boldsymbol{w}_{i,t})},
            \end{align}
            is due to the packet errors, and the second term, i.e., 
            \begin{align}
                \boldsymbol{B}=\frac{\sum_{i=1}^{N} C(\boldsymbol{w}_{i,t}) \sum_{m=1}^{D_i} \nabla_{\boldsymbol{wx}}^2 f(\boldsymbol{w}_t, \boldsymbol{x}_{im}, y_{im}) \boldsymbol{n}_{im}}{D},
            \end{align}
            is affected both by the packet errors and sensor noise, and we have $\mathbb{E}(\|\boldsymbol{J}\|^2) \le \mathbb{E}(\|\boldsymbol{A}\|^2) + \mathbb{E}(\|\boldsymbol{B}\|^2)$. We now derive bounds on the two terms $\mathbb{E}(\|\boldsymbol{A}\|^2)$ and $\mathbb{E}(\|\boldsymbol{B}\|^2)$, respectively. For $\mathbb{E}(\|\boldsymbol{A}\|^2)$ we have 
            \begin{align}
                &\mathbb{E}(\|\boldsymbol{A}\|^2)=\mathbb{E}(\|\nabla F(\boldsymbol{w}_t)-\frac{\sum_{i=1}^{N} \sum_{m=1}^{D_i} \nabla f(\boldsymbol{w}, \boldsymbol{x}_{im}, y_{im})C(\boldsymbol{w}_{i,t})}{\sum_{i=1}^{N}D_i C(\boldsymbol{w}_{i,t})}\|^2) \\ \nonumber
                &=\mathbb{E}(\|-\frac{(D-\sum_{i=1}^{N}D_i C(\boldsymbol{w}_{i,t}))\sum_{i \in E^c} \sum_{m=1}^{D_i} \nabla f(\boldsymbol{w}, \boldsymbol{x}_{im}, y_{im})}{D \sum_{i=1}^{N}D_i C(\boldsymbol{w}_{i,t})}\\ \nonumber
                &+\frac{\sum_{i \in E} \sum_{m=1}^{D_i} \nabla f(\boldsymbol{w}, \boldsymbol{x}_{im}, y_{im})}{D}\|^2) \\ \nonumber
                & \le \mathbb{E}(\frac{(D-\sum_{i=1}^{N}D_i C(\boldsymbol{w}_{i,t}))\sum_{i \in E^c} \sum_{m=1}^{D_i} \|\nabla f(\boldsymbol{w}, \boldsymbol{x}_{im}, y_{im})\|}{D \sum_{i=1}^{N}D_i C(\boldsymbol{w}_{i,t})}\\ \nonumber
                &+\frac{\sum_{i \in E} \sum_{m=1}^{D_i} \|\nabla f(\boldsymbol{w}, \boldsymbol{x}_{im}, y_{im})\|}{D})^2,
            \end{align}
            where $E=\{i|C(\boldsymbol{w}_{i,t})=0\}$ is the set of users that experience packet errors and its complement denoted by $E^c$ is the set of users that correctly transmit their local FL models to the BS. With assumption V, $\|\nabla f(\boldsymbol{w}_t, \boldsymbol{x}_{im}, y_{im})\| \le \sqrt{c_1+c_2\|\nabla F(\boldsymbol{w}_t)\|^2}$, we have $\sum_{i \in E^c} \sum_{m=1}^{D_i} \|\nabla f(\boldsymbol{w}, \boldsymbol{x}_{im}, y_{im})\| \le \sqrt{c_1+c_2\|\nabla F(\boldsymbol{w}_t)\|^2} (\sum_{i=1}^{N} D_i C(\boldsymbol{w}_{i,t})) $ and $\sum_{i \in E} \sum_{m=1}^{D_i} \|\nabla f(\boldsymbol{w}, \boldsymbol{x}_{im}, y_{im})\| \le \sqrt{c_1+c_2\|\nabla F(\boldsymbol{w}_t)\|^2} (D-\sum_{i=1}^{N} D_i C(\boldsymbol{w}_{i,t}))$. Hence, $\mathbb{E}(\|\boldsymbol{A}\|^2)$ can be expressed by
            \begin{align}
                \mathbb{E}(\|\boldsymbol{A}\|^2) \le \frac{4}{D^2} \mathbb{E}\left( D- \sum_{i=1}^{N} D_i C(\boldsymbol{w}_{i,t})\right)^2 (c_1+c_2\|\nabla F(\boldsymbol{w}_t) \|^2).
            \end{align}
            Since $D \ge D- \sum_{i=1}^{N} D_i C(\boldsymbol{w}_{i,t}) \ge 0$, we have
            \begin{align}
                \mathbb{E}(\|\boldsymbol{A}\|^2) \le \frac{4}{D} \mathbb{E}\left( D- \sum_{i=1}^{N} D_i C(\boldsymbol{w}_{i,t})\right) (c_1+c_2\|\nabla F(\boldsymbol{w}_t) \|^2).
            \end{align}
            Since, $\mathbb{E}(C(\boldsymbol{w}_{i,t}))=1-e_{i,t}$, we have
            \begin{align}\label{bound1}
                \mathbb{E}(\|\boldsymbol{A}\|^2) \le \frac{4}{D} \left( \sum_{i=1}^{N} D_i e_{i,t} \right) (c_1+c_2\|\nabla F(\boldsymbol{w}_t) \|^2).
            \end{align}



            
            The second term $\mathbb{E}(\|\boldsymbol{B}\|^2)$ is caused by the sensor noise and can be bounded as 
            

            \begin{align}\label{bound2}
                &\mathbb{E}(\|\boldsymbol{B}\|^2) = \frac{1}{D^2}\mathbb{E}(\|\sum_{i=1}^{N} C(\boldsymbol{w}_{i,t}) \sum_{m=1}^{D_i} \nabla_{\boldsymbol{wx}}^2 f(\boldsymbol{w}_t, \boldsymbol{x}_{im}, y_{im}) \boldsymbol{n}_{im}\|^2) \\ \nonumber
                & = \frac{1}{D^2} \sum_{i}\sum_{i'} \mathbb{E}[C(\boldsymbol{w}_{i,t}) C(\boldsymbol{w}_{i',t})] \sum_{m}\sum_{m'} \mathbb{E}[\boldsymbol{n}_{im}^T \nabla_{\boldsymbol{wx}}^2 f^T \nabla_{\boldsymbol{wx}}^2 f \boldsymbol{n}_{i'm'}] \\ \nonumber
                & = \frac{1}{D^2} \sum_{i=1}^{N} \mathbb{E}[C(\boldsymbol{w}_{i,t})] \sum_{m=1}^{D_i} \mathbb{E}[\boldsymbol{n}_{im}^T \nabla_{\boldsymbol{wx}}^2 f^T \nabla_{\boldsymbol{wx}}^2 f \boldsymbol{n}_{im}] \\ \nonumber
                & \le \frac{\eta}{D^2} \sum_{i=1}^{N} (1-e_{i,t}) \sum_{m=1}^{D_i} \mathbb{E}[\boldsymbol{n}_{im}^T  \boldsymbol{n}_{im}] = \frac{\eta M}{D^2} \sum_{i=1}^{N} D_i (1-e_{i,t}) \sigma^2_i,
            \end{align}

            \begin{figure*}[t]
            \centering
            \begin{minipage}[t]{.40\textwidth}
            \begin{subfigure}{\textwidth}
              \includegraphics[width=\linewidth]{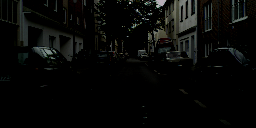}
              \caption{First street scene}
              \label{fig:1_usg}
            \end{subfigure}
            \end{minipage}\hfill
            \begin{minipage}[t]{.40\textwidth}
            \begin{subfigure}{\textwidth}
              \includegraphics[width=\linewidth]{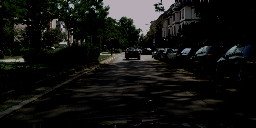}
              \caption{Second street scene}
              \label{fig:2_usg}
            \end{subfigure}
            \end{minipage}\\
            
            \begin{minipage}[t]{.40\textwidth}
            \begin{subfigure}{\textwidth}
              \includegraphics[width=\linewidth]{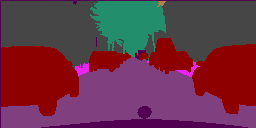}
              \caption{First segmented ground truth}
              \label{fig:1_sg}
            \end{subfigure}
            \end{minipage}\hfill
            \begin{minipage}[t]{.40\textwidth}
            \begin{subfigure}{\textwidth}
              \includegraphics[width=\linewidth]{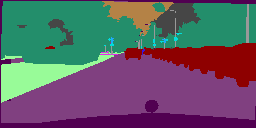}
              \caption{Second segmented ground truth}
              \label{fig:2_sg}
            \end{subfigure}
            \end{minipage}\\
            
            \begin{minipage}[t]{.40\textwidth}
            \begin{subfigure}{\textwidth}
             \includegraphics[width=\linewidth]{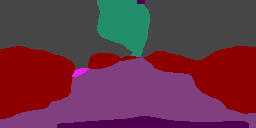}
              \caption{Proposed solution (PA=$88.38\%$, $\mathrm{IoU}(Car)$=$79.20\%$)}
              \label{fig:op_1_gt}
            \end{subfigure}
            \end{minipage}\hfill
            \begin{minipage}[t]{.40\textwidth}
            \begin{subfigure}{\textwidth}
             \includegraphics[width=\linewidth]{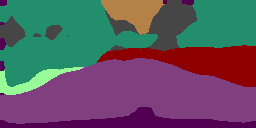}
              \caption{Proposed solution (PA=$82.5\%$, $\mathrm{IoU}(Car)$=$78.39\%$)}
              \label{fig:op_2_gt}
            \end{subfigure}
            \end{minipage}\\

            \begin{minipage}[t]{.40\textwidth}
            \begin{subfigure}{\textwidth}
             \includegraphics[width=\linewidth]{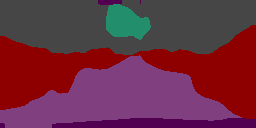}
              \caption{Maximum  rate \cite{Rui} (PA=$82.65\%$, $\mathrm{IoU}(Car)$=$73.81\%$)}
              \label{fig:op_comm_1_gt}
            \end{subfigure}
            \end{minipage}\hfill
            \begin{minipage}[t]{.40\textwidth}
            \begin{subfigure}{\textwidth}
             \includegraphics[width=\linewidth]{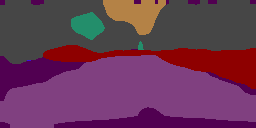}
              \caption{Maximum  rate \cite{Rui} (PA=$55.31\%$, $\mathrm{IoU}(Car)$=$60.42\%$)}
              \label{fig:op_comm_2_gt}
            \end{subfigure}
            \end{minipage}\\
            
            \begin{minipage}[t]{.40\textwidth}
            \begin{subfigure}{\textwidth}
             \includegraphics[width=\linewidth]{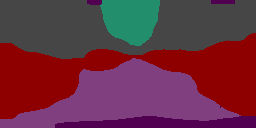}
              \caption{Weighted centroid (PA=$85.62\%$, $\mathrm{IoU}(Car)$=$76.19\%$)}
              \label{fig:wc_1_gt}
            \end{subfigure}
            \end{minipage}\hfill
            \begin{minipage}[t]{.40\textwidth}
            \begin{subfigure}{\textwidth}
             \includegraphics[width=\linewidth]{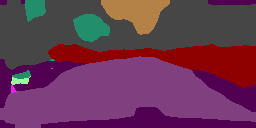}
              \caption{Weighted centroid (PA=$55.45\%$, $\mathrm{IoU}(Car)$=$65.53\%$)}
              \label{fig:wc_2_gt}
            \end{subfigure}
            \end{minipage}\\

            \begin{minipage}[t]{.40\textwidth}
            \begin{subfigure}{\textwidth}
             \includegraphics[width=\linewidth]{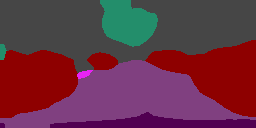}
              \caption{Noise unaware\cite{Joint} (PA=$86.22\%$, $\mathrm{IoU}(Car)$=$78.06\%$)}
              \label{fig:walid_1_gt}
            \end{subfigure}
            \end{minipage}\hfill
            \begin{minipage}[t]{.40\textwidth}
            \begin{subfigure}{\textwidth}
             \includegraphics[width=\linewidth]{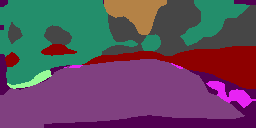}
              \caption{Noise unaware\cite{Joint} (PA=$70.49\%$, $\mathrm{IoU}(Car)$=$63.52\%$)}
              \label{fig:walid_2_gt}
            \end{subfigure}
            
            \end{minipage}\\
            \caption{\color{black}Visual quality and IoU comparisons for collaborative drone-assisted semantic segmentation on autonomous vehicles.}
            \label{fig:cityscape_validation}
            \end{figure*}

            where the inequality stems from assumption IV. Hence, we get
            \begin{align}
                 &\mathbb{E}(F(\boldsymbol{w}_{t+1})) \le \mathbb{E} (F(\boldsymbol{w}_t))+\frac{\eta M}{2LD^2} \sum_{i=1}^{N} D_i (1-e_{i,t}) \sigma^2_i\\ \nonumber
                &+ \frac{2c_1}{L D} \sum_{i=1}^{N} D_i e_{i,t}-\frac{1}{2L} (1- \frac{4c_2}{D} \sum_{i=1}^{N} D_i e_{i,t}) \| \nabla F(\boldsymbol{w}_t)\|^2.
            \end{align}
            Subtracting $\mathbb{E}(F(\boldsymbol{w}^*))$ from both sides, we have
            \begin{align}
                 \mathbb{E}(F(\boldsymbol{w}_{t+1})-F(\boldsymbol{w}^*)) &\le \mathbb{E}(F(\boldsymbol{w}_{t})-F(\boldsymbol{w}^*)) \\ \nonumber 
                 &-\frac{1}{2L} (1- \frac{4c_2}{D} \sum_{i=1}^{N} D_i e_{i,t}) \| \nabla F(\boldsymbol{w}_t)\|^2 \\ \nonumber
                 &+ \frac{2c_1}{L D} \sum_{i=1}^{N} D_i e_{i,t} \\ \nonumber
                 &+\frac{\eta M}{2LD^2} \sum_{i=1}^{N} D_i (1-e_{i,t}) \sigma^2_i.
            \end{align}
            With assumptions II and III, we have $\| \nabla F(\boldsymbol{w}_t)\|^2 \ge 2\mu (F(\boldsymbol{w}_t)-F(\boldsymbol{w}^*))$ \cite{Boyd} which gives
            \begin{align}\label{recur}
                &\mathbb{E}(F(\boldsymbol{w}_{t+1})-F(\boldsymbol{w}^*)) \le \Phi \mathbb{E}(F(\boldsymbol{w}_{t})-F(\boldsymbol{w}^*)) \\ \nonumber 
                 &+ \frac{2c_1}{L D} \sum_{i=1}^{N} D_i e_{i,t}+\frac{\eta M}{2LD^2} \sum_{i=1}^{N} D_i (1-e_{i,t}) \sigma^2_i,
            \end{align}
            where $\Phi_t=1-\frac{\mu}{L}+\frac{4\mu c_2}{LD} \sum_{i=1}^N D_i e_{i,t}$ controls the convergence speed and the error terms $J_t=\frac{2c_1}{L D} \sum_{i=1}^{N} D_i e_{i,t}$ and $K_t=\frac{\eta M}{2L D^2} \sum_{i=1}^{N} D_i (1-e_{i,t}) \sigma^2_i$ are due to the packet errors and sensor noise.
            
            \end{FlushLeft}
        \end{proof}
\end{appendices}



\bibliographystyle{IEEEtran}
\bibliography{Bibliography.bib}

\end{document}